\documentclass{elsarticle}
\usepackage{times}

\usepackage{subfig}
\usepackage{amssymb}
\usepackage{amsmath}
\usepackage{amsthm}
\usepackage{stackengine}
\usepackage{graphicx}
\usepackage{rotating}
\usepackage[ruled, vlined, longend, linesnumbered]{algorithm2e}
\usepackage{tikz}
\usetikzlibrary{arrows,automata}
\usepackage{mathtools}
\usepackage{fancybox}
\usepackage{empheq}
\usepackage{url}

\DeclareMathOperator*{\argmin}{arg\,min}

\newtheorem{thrm}{Theorem}
\newtheorem{prop}{Proposition}
\newtheorem{defi}{Definition}

\newtheorem{exam}{Example}

\newtheorem{obs}{Observation}

\begin{document}

\begin{frontmatter}
\title{A double oracle approach to minmax regret optimization problems with interval data}


\author[1]{Hugo Gilbert\corref{cor1}}
\ead{hugo.gilbert@lip6.fr}

\author[1]{Olivier Spanjaard}
\ead{olivier.spanjaard@lip6.fr}

\cortext[cor1]{Corresponding author}

\address[1]{Sorbonne Universit\'es, UPMC Univ Paris 06, CNRS, LIP6 UMR 7606, 4 place Jussieu, 75005 Paris}

\begin{abstract}
In this paper, we provide a generic anytime lower bounding procedure for minmax regret optimization problems. We show that the lower bound obtained is always at least as accurate as the lower bound recently proposed by Chassein and Goerigk \cite{DBLP:journals/eor/ChasseinG15}. This lower bound can be viewed as the optimal value of a linear programming relaxation of a mixed integer programming formulation of minmax regret optimization, but the contribution of the paper is to compute this lower bound via a double oracle algorithm \cite{McMahanGB03} that we specify. The double oracle algorithm is designed by relying on a game theoretic view of robust optimization, similar to the one developed by Mastin \textit{et al.} \cite{mastin2015randomized}, and it can be efficiently implemented for any minmax regret optimization problem whose standard version is ``easy''. We describe how to efficiently embed this lower bound in a branch and bound procedure. Finally we apply our approach to the robust shortest path problem. Our numerical results show a significant gain in the computation times compared to previous approaches in the literature.     
\end{abstract}

\end{frontmatter}


\section{Introduction} 
\label{sec:introduction}

The definition of an instance of a combinatorial optimization problem requires to specify parameters (e.g., costs of edges in a network problem) wich can be uncertain or imprecise. For instance, instead of specifying scalar values, there are situations where only an interval of possible values is known for each parameter. It is then assumed that the different parameters (e.g., the different costs of edges) may take any value of the given interval independently of the other parameters. The assignment of a scalar value in each interval is called a \emph{scenario}.

One may then optimize the worst case performance (i.e., performance of the considered solution in the worst possible scenario) but this approach often leads to an overly conservative solution. A less conservative approach, which is known as minmax regret optimization, minimizes the maximum difference in the objective value of a solution over all scenarios, compared to the best possible objective value attainable in this scenario. Unfortunately, the minmax regret versions of combinatorial optimization problems with interval data are often NP-hard.

For tackling minmax regret optimization problems, it is therefore useful to investigate efficient algorithms with performance guarantee.
Kasperski and Zielinski proved that the algorithm that returns a midpoint solution (i.e., an optimal solution in the scenario where one takes the middle of each interval) has an approximation ratio of 2 \cite{KasperskiZ06}. 
Recently, Chassein and Goerigk presented a new approach for determining a lower bound on the objective value of a minmax regret solution. It allows them to compute a tighter \emph{instance-dependent} approximation ratio of a midpoint solution, that is of course at most 2 \cite{DBLP:journals/eor/ChasseinG15}. Moreover they showed how their bound can be used in a branch and bound framework for the minmax regret version of the shortest path problem with interval data, and obtained an improvement on the computation times compared to state-of-the-art approaches, especially when the uncertainty is high. Following Chassein and Goerigk, we also present a very general approach to compute a lower bound on the objective value of a minmax regret solution, which further improves the instance-dependent approximation ratio of a midpoint solution.

Our approach relies on a game-theoretic view of robust optimization. Such a game-theoretic view has already been adopted once in the literature by Mastin \textit{et al.} \cite{mastin2015randomized}. They view robust optimization as a two-player zero-sum game where one player (optimizing player) selects a feasible solution, and the other one (adversary) selects a possible scenario for the values of the parameters. More precisely, they consider a randomized model where the optimizing player selects a probability distribution over solutions (called a mixed solution hereafter) and the adversary selects values with knowledge of the player's distribution but not its realization. They show that the determination of a randomized Nash equilibrium of this game (and thus of a minimax regret mixed solution for the optimizing player) can be performed in polynomial time provided the standard version of the problem (i.e., where the precise values of the parameters are known) is itself polynomial. They use a linear programming formulation that may have an exponential number of constraints, but for which a polynomial time separation oracle amounts to solving an instance of the standard problem.

In this article, we also compute a randomized Nash equilibrium of the game but we use a double oracle approach \cite{McMahanGB03} that reveals more efficient in computation times. Furthermore, we show how the bound resulting from this computation can be used in a branch and bound algorithm for the minmax regret version of a (deterministic) robust optimization problem. We illustrate the interest of this approach on the robust shortest path problem with interval data. Comparing our results with the ones obtained by Chassein and Goerigk \cite{DBLP:journals/eor/ChasseinG15} and Montemanni et al. \cite{DBLP:journals/orl/MontemanniGD04} we find considerable improvements in computation times on some classes of instances.

The remainder of the paper is structured as follows. Section 2 recalls the definition of a minmax regret problem and introduces the notations. In section 3, a game theoretic view of minmax regret optimization, similar to the one proposed by Mastin \textit{et al.} \cite{mastin2015randomized}, is developed. This view helps us derive a precise lower bound on the objective value of a minmax regret solution. The method to compute this lower bound uses a double oracle algorithm described in section 4. In section 5 we discuss how our work can be used efficiently in a branch and bound algorithm. Lastly, in section 6, we test our approach on the mimax regret version of the shortest path problem with interval data.  

\section{Minmax regret problems}
\label{sec:minmaxRegret}
To study minmax regret problems, we adopt the notations of Chassein and Goerigk \cite{DBLP:journals/eor/ChasseinG15}. A combinatorial optimization problem is stated as follows (in minimization):
\begin{equation}
 \underset{x\in\mathcal{X}}{\min} \sum_{i=1}^n c_i x_i
\label{eq:P}
\end{equation}
where $\mathcal{X} \subset \{0,1\}^n =: \mathbb{B}^n$ denotes the set of feasible solutions, and $c_i$ represents the cost of element $i$. We refer to (\ref{eq:P}) as the standard optimization problem.

In robust optimization, instead of considering only a single cost vector $c$ we assume there is an uncertainty set $\mathcal{U}\subset \mathbb{R}^n$ of different possible cost vectors $c$. Each possible cost vector $c$ is also called a scenario. Two approaches can
be distinguished according to the way the set of scenarios is defined: the \emph{interval model} where each $c_i$ takes value in an interval $[\underline{c}_i,\overline{c}_i]$ and where the set of scenarios is defined implicitly as the Cartesian product $\mathcal{U} = \times_{i=1}^{n}[\underline{c}_i,\overline{c}_i]$; the \emph{discrete scenario model} where the possible scenarios consist in a finite set of cost vectors. Intuitively, a \emph{robust} solution is a solution that remains suitable whatever scenario finally occurs. Several criteria have been proposed to formalize this: the \emph{minmax} criterion consists of evaluating a solution on the basis of its worst value over all scenarios, and
the \emph{minmax regret} criterion consists of evaluating a solution on the basis of its maximal deviation from the optimal value over all scenarios. 


We consider here the approach by minmax regret optimization in the interval model. For a feasible solution $x$, the regret $Reg (x,c)$ induced by a scenario $c$ is defined as the difference between the realized objective value $val(x, c)$ and the best possible objective value $\min_y val(y, c)$ in this scenario, where $val(x, c)=\sum_{i=1}^n c_i x_i$. 
Similarly, for a feasible solution $x$, the regret $Reg (x,y)$ induced by another feasible solution $y$ is defined by maximizing over possible scenarios $c\in\mathcal{U}$ the difference $Reg (x,y,c)$ between the realized objective values $val(x, c)$ and $val(y, c)$.
The (maximal) regret $Reg(x)$ of a feasible solution $x$ is then defined as the maximal regret induced by a scenario or equivalently the maximal regret induced by another solution. More formally, we have:
\begin{align}
  Reg (x,y,c) &= val(x, c) - val(y, c) \\ 
  Reg (x,c) &= \max_{y\in\mathcal{X}} Reg(x,y,c) = val(x, c) - val(x^c, c) \\ 
  Reg (x,y) &= \max_{c\in\mathcal{U}} Reg(x,y,c) \\ 
  Reg (x) &= \underset{c\in\mathcal{U}}{\max} Reg(x,c) = \underset{y\in\mathcal{X}}{\max} Reg(x,y) 
\end{align}
where $x^c = \argmin_{x\in\mathcal{X}} \sum_{i=1}^n c_i x_i$ is an optimal solution of the standard optimization problem for cost vector $c$. In the following, for brevity, we may denote by $val^{*}(c) = val(x^c, c)$ the objective value of an optimal solution of the standard optimization problem for cost vector $c$.

In a minmax regret optimization problem, the goal is to find a feasible solution with minimal regret.
The minmax regret optimization problem 
can thus be formulated as:

\begin{equation}
\underset{x\in\mathcal{X}}{\min}Reg(x)
\label{eq:R}
\end{equation}


The optimal value for this problem is denoted by $OPT$ in the remainder of the article. As stated in the introduction, most minmax regret versions of standard optimization problems are NP-hard~\cite{aissi2009min}. It is therefore worth investigating approximation algorithms. To the best of our knowledge, the most general polynomial time approach in this concern is the midpoint algorithm proposed by Kasperski and Zielinski \cite{KasperskiZ06} which returns a solution minimizing $\sum_{i=1}^n \hat{c}_i x_i$, where $\hat{c}_i = (\underline{c}_i + \overline{c}_i)/2$. We denote by $x_{mid}$ such a solution.

The regret of a midpoint solution $x_{mid}$ is always not more than $2.OPT$ \cite{KasperskiZ06}. While there are problem instances where this approximation guarantee is tight, the regret is actually smaller for a lot of instances. In order to account for that, one needs to design a tighter \emph{instance-dependent} ratio. In this line of research, Chassein and Goerigk \cite{DBLP:journals/eor/ChasseinG15} show how to improve the 2-approximation guarantee for every specific instance with only small computational effort. To this end they propose a new lower bound $LB_{CG}$ on the minmax regret value. This lower bound is then used to compute an instance-dependent approximation guarantee $Reg(x_{mid})/LB_{CG}$ for a midpoint solution $x_{mid}$. This guarantee equals 2 in the worst case. 


In this paper we investigate another lower bound based on game-theoretic arguments. This lower bound corresponds to the value of a randomized Nash equilibrium in a game that will be specified in the sequel. We prove that this lower bound is tighter than the one presented by Chassein and Goerigk \cite{DBLP:journals/eor/ChasseinG15}. Therefore, this lower bound further improves the instance-dependent approximation guarantee of a midpoint solution.

Note that Mastin \textit{et al.} investigated a similar game-theoretic view of robust optimization \cite{mastin2015randomized}. They showed how to compute a randomized Nash equilibrium by using a mathematical programming formulation involving an exponential number of constraints. They propose to solve this program via a cutting-plane method relying on a polynomial time separation oracle provided the standard version of the tackled problem is itself polynomial. By using the ellipsoid method, it follows that the complexity of computing a randomized Nash equilibrium is polynomial by the polynomial time equivalence of the separation problem and the optimization problem \cite{Grotschel81} (but no numerical tests have been carried out in their paper). We consider here the same game but we propose a \emph{double oracle} approach to compute a randomized Nash equilibrium. While the proposed double oracle approach is not a polynomial time algorithm, it will reveal more efficient in practice than a cutting plane approach.

In the next section, we present the game-theoretic view of robust combinatorial optimization. This will allow us to define the lower bound considered in this paper.

\section{Game Theoretic View}
\label{sec:game}
The minmax regret optimization problem defined by (\ref{eq:R}) induces a zero-sum two-player game where the sets of pure strategies are the set of feasible solutions $\mathcal{X}$ for player 1 (also called $x$-player) and the set of scenarios $\mathcal{U}$ for player 2 (also called $c$-player). The resulting payoff is then given by $Reg(x,c)$, the regret of using solution $x$ in scenario $c$. Obviously, the $x$-player (resp $c$-player) aims at minimizing (resp. maximizing) $Reg(x,c)$. The $x$-player (resp. $c$-player) can also play a mixed strategy by playing according to a probability distribution $P_{\mathcal{X}}$ (resp. $P_{\mathcal{U}}$) over set $\mathcal{X}$ (resp. $\mathcal{U}$). We denote by $P_{\mathcal{X}}(x)$ (resp. $P_{\mathcal{U}}(c)$) the probability to play strategy $x$ (resp. $c$) in mixed strategy $P_{\mathcal{X}}$ (resp. $P_{\mathcal{U}}$). In the following, we restrict ourselves to mixed strategies over a \emph{finite} set of pure strategies. We will show at the end of this section that this assumption is not restrictive.
We denote by $\Delta_{\mathcal{X}}$ (resp. $\Delta_{\mathcal{U}}$) the set of all possible mixed strategies for the $x$-player (resp. $c$-player). We follow the convention to denote random variables with capital letters (e.g., we write $C$ instead of $c$ to denote random cost vectors). 

The regret functions $Reg(.), Reg(.,.), Reg(.,.,.)$ are then extended to mixed strategies by linearity in probability. For instance, if the $c$-player plays $P_{\mathcal{U}}$, the regret of the $x$-player for playing $P_{\mathcal{X}}$  is the expectancy $\sum_{x \in \mathcal{X}}\sum_{c \in \mathcal{U}} P_{\mathcal{X}}(x) P_{\mathcal{U}}(c) Reg(x,c)$.  

\subsection{Best response functions and lower bound}
Given the strategy of one player, a best reponse is an optimal pure strategy that can be played by the other player. We will call a best $x$-response (resp. $c$-response) a best response of the $x$-player (resp. $c$-player). We can now state some results on the best responses of each player.    

\begin{obs}
A best $x$-response to a scenario $c$ is given by $x^c$.
\end{obs}

The next observation is a well known characterization of a worst scenario for a given feasible solution and has been used in many related papers \cite{KPY01,aissi2009min}.
\begin{obs}
A best $c$-response to a feasible solution $x$ is given by $c^x$ defined by $c_i^x = \overline{c}_i$ if $x_i=1$ and $\underline{c}_i$ otherwise. \label{obs:cresp-sol}
\end{obs}

The scenario $c^x$ belongs to the subset of \emph{extreme scenarios}: a scenario $c$ is said to be an extreme scenario if $\forall i \in \{1,\ldots,n\}$, $c_i = \underline{c}_i$ or $c_i =  \overline{c}_i$. Given an extreme scenario $c$, we will denote by $\lnot c$ the opposite extreme scenario defined by $\lnot {c}_i = \underline{c}_i$ if $ c_i = \overline{c}_i$ and $\overline{c}_i$ otherwise. While $c^x$ is the most penalizing scenario for solution $x$, $\lnot {c^x}$ is the most favorable scenario for solution $x$. 

The following observation will reveal useful in the proof of Proposition~\ref{propBestCResponse} below.


\begin{obs} \label{obs:regy}
Given two feasible solutions $x$ and $y$ in $\mathcal{X}$, $Reg(x, y)$ is achieved by both $c^x$ and $\neg c^y$: $Reg(x,y) = Reg(x,y,c^x) =Reg(x,y,\neg c^y)$.
\end{obs}

Observation~\ref{obs:xOracle} gives us a convenient way to determine a best response for the $x$-player against a mixed strategy of the $c$-player.

\begin{obs}
A best $x$-response to a mixed strategy $P_{\mathcal{U}}$ is given by $x^{\tilde{c}}$ where $\tilde{c}$ is defined by $\tilde{c} = \sum_{c\in \mathcal{U}} P_{\mathcal{U}}(c) c$. \label{obs:xOracle}
\end{obs}
\begin{proof}
It follows from the following sequence of equalities:
\begin{align*}
\min_x Reg(x,P_{\mathcal{U}}) &= \min_x \left(\sum_{c\in \mathcal{U}} P_{\mathcal{U}}(c)(val(x,c) - val^{*}(c))\right)\\
                               &= \min_x \left(\sum_{c \in \mathcal{U}}P_{\mathcal{U}}(c)val(x,c)\right) - \sum_{c \in \mathcal{U}}P_{\mathcal{U}}(c)val^{*}(c)\\
                               &= \min_x \left(val(x,\sum_{c \in \mathcal{U}}P_{\mathcal{U}}(c)c)\right) - \sum_{c \in \mathcal{U}}P_{\mathcal{U}}(c)val^{*}(c)
\end{align*}
In the last line, the second term does not depend on $x$. By definition, the solution that minimizes the first term is $x^{\tilde{c}}$.
\end{proof}

Conversely, determining a best response for the $c$-player against a mixed strategy of the $x$-player is slightly more involved. 

To give the intuition of the way a best $c$-response (i.e., a worst scenario) to a mixed strategy can be characterized, it is useful to come back to the case of a best $c$-response to a pure strategy. Given a feasible solution $x$, the worst scenario $c^x$ in Observation~\ref{obs:cresp-sol} is the one that \emph{penalizes} most $x$. Actually, one can consider a dual viewpoint where one \emph{favors} most the solution $y$ that will induce the max regret for $x$. This corresponds to scenario $\neg c^y$, and it is another worst scenario for $x$.
For the convenience of the reader, we illustrate this point on a simple example.

\begin{exam}
Consider the very simple minmax regret optimization problem defined by $n=5$, $c_1 \in [3,4]$, $c_2 \in [1,5]$, $c_3 \in [1,2]$ , $c_4 \in [3,3]$, $c_5 \in [0,6]$ and $\mathcal{X} = \{x \in \mathbb{B}^5 : \sum_{i=1}^5 x_i = 2\}$. Let $x$ be a feasible solution defined by $x_2 = x_3 = 1$ and $x_i = 0$ for $i \not\in\{2,3\}$. The worst scenario $c^x$ as defined in Observation~\ref{obs:cresp-sol} is obtained for $c_1 = 3$, $c_2 = 5$, $c_3 = 2$ , $c_4 = 3$ and $c_5 = 0$. The best feasible solution $x^{c^x}$ in scenario $c^x$ is defined by $x_3 = x_5 = 1$ and $x_i = 0$ for $i \not\in\{3,5\}$. It is then easy to see that the scenario $c$ defined by $c_i = \underline{c}_i$ if $x^{c^x}_i = 1$ and $\overline{c}_i$ otherwise yields the same regret for $x$, and is therefore also a worst scenario for $x$. In the small numerical example, it corresponds to $c_1 = 4$, $c_2 = 5$, $c_3 = 1$, $c_4 = 3$ and $c_5 = 0$. 
\end{exam}

Now, this idea can be generalized to characterize a best $c$-response to a mixed strategy $P_{\mathcal X}$, by looking for a feasible solution $y$ that maximizes $Reg(P_{\mathcal{X}},y)$:
\begin{enumerate}
\item identify a specific scenario $\tilde{c}^{P_{\mathcal{X}}}$ (to be precisely defined below, and that amounts to $c^x$ if $P_{\mathcal{X}}$ chooses $x$ with probability 1) that takes into account the expected value of $x_i$ ($i=1,\ldots,n$) given $P_{\mathcal X}$;
\item compute the best feasible solution $x^{\tilde{c}^{P_{\mathcal{X}}}}$ regarding this scenario; this is the above-mentioned solution $y$ that will induce the max regret for $P_{\mathcal{X}}$;
\item the worst scenario $c^{P_{\mathcal{X}}}$ for $P_{\mathcal X}$ is defined by $c^{P_{\mathcal{X}}}_i = \underline{c}_i$ if $x^{\tilde{c}^{P_{\mathcal{X}}}}_i = 1$ and $\overline{c}_i$ otherwise; this is indeed a scenario that maximizes the regret of choosing $P_{\mathcal{X}}$ instead of $x^{\tilde{c}^{P_{\mathcal{X}}}}$.
\end{enumerate}

This principle can be formalized as indicated in Proposition~\ref{propBestCResponse} below. Note that this result was first stated by Mastin \emph{et al.} \cite{mastin2015randomized}. For the sake of completeness, we provide a proof of the result.

\begin{prop} \label{propBestCResponse} \emph{\textbf{\cite{mastin2015randomized}}}
Given a probability distribution $P_{\mathcal{X}}$, let $x^{\tilde{c}^{P_{\mathcal{X}}}}$ be an optimal solution in scenario $\tilde{c}^{P_{\mathcal{X}}}$ defined by $\tilde{c}^{P_{\mathcal{X}}}_i = \underline{c}_i + (\overline{c}_i - \underline{c}_i)\sum_{x\in\mathcal{X}}x_i P_\mathcal{X}(x)$.
A best $c$-response to $P_{\mathcal{X}}$ is given by extreme scenario $c^{P_{\mathcal{X}}}$ where $c^{P_{\mathcal{X}}}$ is defined by $c^{P_{\mathcal{X}}}_i = \underline{c}_i$ if $x_i^{\tilde{c}^{P_{\mathcal{X}}}}=1$ and $\overline{c}_i$ otherwise. \label{prop:cOracle}
\end{prop}

\begin{proof}

Instead of looking for an extreme scenario $c$ that maximizes $Reg(P_{\mathcal{X}},c)$, we look for a feasible solution $y$ that maximizes $Reg(P_{\mathcal{X}},y)$ or equivalently $Reg(P_{\mathcal{X}}, y, \lnot c^{y})$ (due to Observation~\ref{obs:regy}). 
 
The following analysis shows how to compute such a solution.  
\begin{align*}
Reg(P_{\mathcal{X}},y,\lnot c^{y}) &= \sum_{x\in\mathcal{X}} P_{\mathcal{X}}(x)(\sum_{i = 1}^n (x_i - y_i) \lnot c^{y}_i)\\
                                 &= \sum_{i = 1}^n ((\sum_{x\in\mathcal{X}} P_{\mathcal{X}}(x) x_i) - y_i) \lnot c^{y}_i\\
                                 &= \sum_{y_i = 1} ((\sum_{x\in\mathcal{X}} P_{\mathcal{X}}(x) x_i) - 1) \underline{c}_i\\
                                 &+ \sum_{y_i = 0} (\sum_{x\in\mathcal{X}} P_{\mathcal{X}}(x) x_i) \overline{c}_i
\end{align*}
Therefore if we denote by $D_i$ the contribution of having $y_i = 1$ (i.e. how having $y_i = 1$ impacts $Reg(P_{\mathcal{X}},y,\lnot c^{y})$) then:
\begin{align*}
D_i &= ((\sum_{x\in\mathcal{X}} P_{\mathcal{X}}(x) x_i) - 1) \underline{c}_i - (\sum_{x\in\mathcal{X}} P_{\mathcal{X}}(x) x_i) \overline{c}_i\\
         &= -( \underline{c}_i + (\overline{c}_i - \underline{c}_i)(\sum_{x\in\mathcal{X}} P_{\mathcal{X}}(x) x_i))
\end{align*}
Hence:
\begin{align*}
Reg(P_{\mathcal{X}},y,\lnot c^{y}) &= \sum_{y_i = 1} ((\sum_{x\in\mathcal{X}} P_{\mathcal{X}}(x) x_i) - 1) \underline{c}_i\\
                                 &+ \sum_{y_i = 0} (\sum_{x\in\mathcal{X}} P_{\mathcal{X}}(x) x_i) \overline{c}_i\\
                                 &= \sum_{y_i = 1} D_i
                                 + \sum_{i = 1}^n (\sum_{x\in\mathcal{X}} P_{\mathcal{X}}(x) x_i) \overline{c}_i
\end{align*}
As the second term of the sum does not depend on $y$ our optimization problem amounts to find $y$ that maximizes $\sum_{y_i = 1} D_i$
or equivalently that minimizes $\sum_{y_i = 1} -D_i$ which is exactly a standard optimization problem with costs $-D_i$.
\end{proof}



We will now use the following proposition from Chassein and Goerigk to show how the notion of best response can help us in designing a lower bound on $OPT$. 

\begin{prop}\textbf{\emph{\cite{DBLP:journals/eor/ChasseinG15}}}
Let $P_\mathcal{U}$ be a mixed strategy of the $c$-player, we have:
\begin{equation}
OPT \geq \underset{x\in\mathcal{X}}{\min} Reg(x, P_{\mathcal{U}})
\end{equation}
\label{pr:bound}
\end{prop}
\begin{proof}
See Lemma 3.2 in \cite{DBLP:journals/eor/ChasseinG15}.
\end{proof}

In other words, the expected regret of the $x$-player when playing a best response to a mixed $c$-strategy is always a lower bound on $OPT$.
The lower bound we study in this paper, $LB^*$, is the best possible lower bound of the type described by proposition \ref{pr:bound}.  More formally :
\begin{equation}
    LB^* = \underset{P_{\mathcal{U}} \in \Delta_{\mathcal{U}}}{\max} \underset{x\in \mathcal{X}}{\min} Reg(x,P_{\mathcal{U}})
\end{equation}

\subsection{Relation with Chassein and Goerigk's bound}
\label{sec:secCG}
Let us denote by $LB_{CG}$ the lower bound designed by Chassein and Goerigk. 
In order to compare $LB_{CG}$ with $LB^*$, we now provide a game theoretic interpretation of $LB_{CG}$.

\begin{defi}
A mixed $c$-strategy $P_{\mathcal{U}}$ is said to be \emph{centered} if:
\begin{equation}
\forall i \in \{1,\ldots,n\}, \sum_{c\in\mathcal{U}} P_\mathcal{U}(c)c_i = \hat{c}_i = \frac{\underline{c}_i + \overline{c}_i}{2}
\end{equation}
\end{defi}

Put another way, the previous definition means that, in a centered strategy, the expected cost of each element $i$ is the arithmetic mean of $\underline{c}_i$ and $\overline{c}_i$.

\begin{defi}
The \emph{support} of a mixed strategy $P$ is the set of pure strategies that are played by $P$ with a non-zero probability.
\end{defi}

The lower bound given by Chassein and Goerigk \cite{DBLP:journals/eor/ChasseinG15} is the best possible lower bound of the type described by proposition \ref{pr:bound} when the $c$-player only considers the restricted set $\hat{\Delta}_\mathcal{U}$ of centered mixed strategies with two equally likely extreme scenarios $c$ and $\lnot c$ as a support (which guarantees that the mixed strategy is centered because $(c_i + \lnot c_i)/2$ $=$ $(\underline{c}_i + \overline{c}_i)/2$). We can now write more formally the lower bound designed by Chassein and Goerigk as:
\begin{align}
  LB_{CG} = \underset{P_{\mathcal{U}} \in \hat{\Delta}_{\mathcal{U}}}{\max} \underset{x\in \mathcal{X}}{\min} Reg(x,P_{\mathcal{U}})
\end{align}

As $\hat{\Delta}_{\mathcal{U}} \subset \Delta_{\mathcal{U}} $, $LB^*$ will always be at least as accurate as $LB_{CG}$, i.e. $LB^* \ge LB_{CG}$. 
Another direct consequence of this observation is that $LB^* \ge Reg(x^{\hat{c}})/2$ because Chassein and Goerigk proved that:
\[
\frac{OPT}{LB_{CG}} \le \frac{Reg(x^{\hat{c}})}{LB_{CG}} \le 2
\]
We now provide a very simple example where $LB^* > LB_{CG}$.

\begin{exam}
Consider the minmax regret optimization problem defined by $n=2$, $c_1 \in [5,10]$, $c_2 \in [7,12]$, and $\mathcal{X} = \{x \in \mathbb{B}^2 : x_1 + x_2 = 1\}$, where $c_1 x_1 + c_2 x_2$ is the value (to minimize) of a feasible solution $x$ in scenario $c$. There are two feasible solutions $x = (1,0)$ and $x = (0,1)$, and four extreme scenarios $c=(5,12)$, $c=(10,7)$, $c=(5,7)$ and $c=(10,12)$, where $x=(a,b)$ (resp. $c=(a,b)$) means that $x_1=a, x_2=b$ (resp. $c_1=a, c_2=b$).

In this problem, set $\hat{\Delta}_\mathcal{U}$ consists of two mixed strategies, namely:
\begin{itemize}
\item strategy $A$: play $c=(5,12)$ with probability $1/2$, and play $c=(10,7)$ with probability $1/2$;
\item strategy $B$: play $c=(5,7)$ with probability $1/2$, and play $c=(10,12)$ with probability $1/2$.
\end{itemize} 
Regardless of which strategy $A$ or $B$ is played by the $c$-player, the best response for the $x$-player is to play $x=(1,0)$. This yields an expected regret of $(1/2) \times 0 + (1/2) \times 3 = 3/2$ (resp. $0$) for the $x$-player if strategy $A$ (resp. $B$) is played by the $c$-player. Therefore, for this problem instance, $LB_{CG} = \max\{3/2,0\} = 3/2$. 

Now, for bound $LB^*$, consider the (non-centered) mixed strategy consisting in playing scenario $c = (5,12)$ with probability 0.3  and scenario $c=(10,7)$ with probability 0.7. This mixed strategy yields an expected regret of $2.1$ regardless of the fact that the $x$-player plays strategy $x= (1,0)$ (expected regret of $0.3 \times 0 + 0.7 \times 3$) or strategy $x = (0,1)$ (expected regret of $0.3 \times 7 + 0.7 \times 0$). Thus $LB^*\ge 2.1 > LB_{CG}$.

\end{exam}

\subsection{Relation with Nash equilibrium}

We show here that determining $LB^*$ amounts to identifying a mixed Nash equilibrium of the game induced on feasible solutions and scenarios. First, it is well-known in game theory that $\min_{x \in \mathcal{X}} Reg(x,P_{\mathcal{U}}) = \min_{P_{\mathcal{X}} \in \Delta_{\mathcal{X}}} Reg(P_{\mathcal{X}},P_{\mathcal{U}})$ (i.e., there always exists a best response that is a pure strategy). Therefore our lower bound can be rewritten as:
\[
LB^* = \underset{P_{\mathcal{U}} \in \Delta_{\mathcal{U}}}{\max} \underset{P_{\mathcal{X}} \in \Delta_{\mathcal{X}}}{\min} Reg(P_{\mathcal{X}},P_{\mathcal{U}})
\]
A mixed Nash equilibrium is a couple $(P_{\mathcal{X}},P_{\mathcal{U}})$ such that:
\[
\underset{P_{\mathcal{U}} \in \Delta_{\mathcal{U}}}{\max} \underset{P_{\mathcal{X}} \in \Delta_{\mathcal{X}}}{\min} Reg(P_{\mathcal{X}},P_{\mathcal{U}}) = \underset{P_{\mathcal{X}} \in \Delta_{\mathcal{X}}}{\min} \underset{P_{\mathcal{U}} \in \Delta_{\mathcal{U}}}{\max} Reg(P_{\mathcal{X}},P_{\mathcal{U}})
\]
Thus, the lower bound $LB^*$ corresponds to the value of a mixed Nash Equilibrium (NE). The next theorem states that a mixed Nash equilibrium always exists in the game induced by feasible solutions and scenarios.

\begin{thrm}
There exists a (possibly mixed) Nash equilibrium in the zero-sum two-person game induced by the feasible solutions and the possible scenarios.
\end{thrm}
\begin{proof}
Consider the game where the pure strategies of the $c$-player are restricted to the extreme scenarios. The number of feasible solutions and extreme scenarios is finite. Therefore, the game is finite. In a finite zero-sum two-player game, the von Neumann's minimax theorem insures that there exists a (possibly mixed) Nash equilibrium. It is then easy to realize that this NE is also an NE of the original game. Indeed, as a best $c$-response can always be found as an extreme scenario, the $c$-player will not find a better response in the original game than the one she is playing in the restricted game.  
\end{proof}

Determining such a mixed Nash equilibrium can be done by linear programming \cite{Chvatal83}. Indeed, consider a zero-sum two-person game where player 1 (resp. player 2) has $k$ (resp. $l$) pure strategies and denote by $A_{ij}$ the payoff of player 1 when strategies $i$ and $j$ are played respectively by player 1 and player 2. Then, a mixed Nash equilibrium of this game can be determined by solving the following Linear Program (LP), where $p_i$ denotes the probability that player 1 plays pure strategy $i$:
\begin{align}
& \min_{v,p_1,\ldots,p_k} v \notag \\
v &\geq \sum_{i=1}^k p_i A_{ij} \qquad  \forall j \in \{1,\ldots,l\} \label{cPureStrat}\\
 \sum_{i=1}^k p_i &= 1 \notag \\
p_i &\geq 0 \qquad \forall i \in \{1,\ldots,k\} \notag
\end{align}

Constraints~(\ref{cPureStrat}) insure that $v$ is the maximum payoff that player 2 can achieve with her pure strategies (and therefore with \emph{any} strategy, because there always exists a best response that is a pure strategy) if player 1 plays the mixed strategy induced by $p_1,\ldots,p_k$. An optimal mixed strategy for player 1 is then given by the optimal values of variables $p_1,\ldots,p_k$ (an optimal strategy for player 2 is given by the optimal dual variables) and $v$ gives the value of the game (which is $LB^*$ in our case). 

Nevertheless, for the game considered here between the $x$-player and the $c$-player, the complete linear program, that involves as many variables as there are pure strategies for both players, would be huge, and it is therefore not worth considering its generation in extension. To tackle this issue, Mastin \textit{et al.} \cite{mastin2015randomized} relies on a cutting-plane method.

To give a better insight of their approach it reveals useful to revise the definition of the set of pure strategies of the $c$-player. Note that, when defining the game between the $x$-player and the $c$-player, the only pure strategies that matter are those that are best responses to some mixed strategy of the adversary. Thus, an important consequence of Observation~\ref{obs:regy} is that the set of pure strategies of the $c$-player can be restricted to the set $\{\neg c^y:y \in \mathcal{X}\}$, i.e. the $c$-player chooses a solution $y$ and set $c_i = \underline{c}_i$ if $y_i=1$, and $c_i = \overline{c}_i$ otherwise\footnote{Actually, it can even be restricted to the set $\{\neg c^y:y \in \mathcal{X} \mbox{ and } y \mbox{ optimal for } \neg c^y\}$.}. Put another way, the ``relevant'' pure strategies of the $c$-player can be directly characterized as feasible solutions $y$ (each feasible solution $y$ corresponds to scenario $\neg c^y$).

 
Given a feasible solution $y \in \mathcal{X}$, the corresponding scenario $\neg c^y$ is compactly defined by $c_i = \underline{c}_i y_i + \overline{c}_i (1-y_i)$ for $i \in \{1,\ldots,n\}$. Hence, 
\begin{align*}
Reg(P_{\mathcal{X}},y) &= \sum_{i=1}^n (\underline{c}_i y_i + \overline{c}_i (1-y_i)) t_i - \sum_{i=1}^n \underline{c}_i y_i\\
                       &= \sum_{i=1}^n \overline{c}_i t_i - \sum_{i=1}^n  y_i(\underline{c}_i + t_i(\overline{c}_i-\underline{c}_i)) 
\end{align*}
where $t_i  = \sum_{x \in \mathcal{X}} x_i P_{\mathcal{X}}(x)$ denotes the probability that element $i$ belongs to the solution that is drawn according to $P_{\mathcal{X}}(x)$. Note that equation $t_i  = \sum_{x \in \mathcal{X}} x_i P_{\mathcal{X}}(x)$ defines a mapping from $\Delta_{\mathcal{X}}$ to the convex hull $CH(\mathcal{X})$ of $\mathcal{X}$, defined by:
\[
CH(\mathcal{X}) = \left\{\displaystyle\sum_{x \in \mathcal{X}} p_x x : p_x \ge 0 \mbox{ and } \displaystyle\sum_{x \in \mathcal{X}} p_x=1\right\}.
\]
Conversely, from any vector $t$ in the convex hull of $\mathcal{X}$, it is possible to compute in polynomial time a mixed strategy $P_{\mathcal{X}}$ such that  $\sum_{x \in \mathcal{X}} x_i P_{\mathcal{X}}(x) = t_i$ \cite{mastin2015randomized}. 

This observation led Mastin \textit{et al.} to consider the following LP for Robust Optimization (RO) \cite{mastin2015randomized}, where constraints~(\ref{eq:cstexpRO}) are the analog of constraints~(\ref{cPureStrat}) above:
\begin{empheq}[left=\mathcal{P}_{RO}\empheqlbrace]{align}
& \min_{v,t_1,\ldots,t_n} v  \notag \\
v \ge \sum_{i=1}^n \overline{c}_i t_i & - \sum_{i=1}^n  y_i(\underline{c}_i + t_i(\overline{c}_i-\underline{c}_i)) \qquad  \forall y \in \mathcal{X} \label{eq:cstexpRO} \\
t & \in  CH(\mathcal{X}) \notag 
\end{empheq}

This formulation takes advantage of the fact that if one can optimize over $\mathcal{X}$ in polynomial time, then one can separate over $CH(\mathcal{X})$ in polynomial time by the polynomial time equivalence of separation and optimization \cite{Grotschel81}. Furthermore, even if the number of constraints in (\ref{eq:cstexpRO}) may be exponential as there is one constraint per feasible solution $y\in \mathcal{X}$, it follows from Proposition \ref{propBestCResponse} that these constraints can be handled efficiently by using a separation oracle computing an optimal solution in scenario defined by $c_i = \underline{c}_i + t_i(\overline{c}_i-\underline{c}_i)$. This separation oracle is also polynomial time if one can optimize over $\mathcal{X}$ in polynomial time. As noted by Mastin et al. \cite{mastin2015randomized}, it implies that the complexity of determining a mixed NE is polynomial if the standard version of the tackled problem is polynomially solvable.

An equivalent formulation to $\mathcal{P}_{RO}$, that turned out to be more efficient in practice in our numerical tests, reads as follows, where $\pi$ is the value of an optimal solution in the scenario defined by $c_i = \underline{c}_i + t_i(\overline{c}_i-\underline{c}_i)$:
\begin{empheq}[left=\mathcal{P}_{RO}'\empheqlbrace]{align}
& \min_{\pi,t_1,\ldots,t_n} \sum_{i=1}^n \overline{c}_i t_i - \pi  \notag \\
\sum_{i=1}^n  y_i(\underline{c}_i& + t_i(\overline{c}_i-\underline{c}_i)) \geq \pi  \qquad  \forall y \in \mathcal{X} \label{eq:cstexp} \\
t &\in CH(\mathcal{X}) \notag 
\end{empheq}

In the next section, we will present another algorithm able to solve the game induced by the feasible solutions and the possible scenarios. While the method of Mastin \textit{et al.} \cite{mastin2015randomized} dynamically generates the strategies of the $c$-player (strategies $\neg c^y$ for $y \in \mathcal{X}$) and uses an implicit representation for the strategies of the $x$-player (a mixed strategy $P_\mathcal{X}$ is induced from $t \in CH(\mathcal{X})$), our algorithm dynamically generates both the strategies of the $c$-player and the ones of the $x$-player. As we will see, our alternative method is much more efficient on some classes of instances.



\section{Solving the Game} 
\label{sec:optimal}

The game can be solved by specifying a double oracle algorithm \cite{McMahanGB03} adapted to our problem.

\textbf{Double oracle approach.} The double oracle algorithm finds a Nash equilibrium for a finite zero-sum two player game where a best response procedure (also called oracle) exists for each player. 
Given a mixed strategy $P_\mathcal{X}$ (resp. $P_\mathcal{U}$), $BR_c(P_\mathcal{X})$ (resp. $BR_x(P_\mathcal{U})$) returns a pure strategy $c$ (resp. $x$) that maximizes $Reg(P_\mathcal{X},c)$ (resp. minimizes $Reg(x,P_\mathcal{U})$).
Those two procedures directly follow from Proposition \ref{prop:cOracle} and Observation \ref{obs:xOracle}. Each procedure only requires to solve one standard optimization problem.
The algorithm starts by considering only small subsets $S_x$ and $S_c$ of pure strategies (singletons in Algorithm~\ref{alg:doa1}) for the $x$-player and the $c$-player, and then grows those sets in every iteration 
by applying the best-response oracles to the optimal strategies (given by the current NE) the players can play in the restricted game $G=(S_x,S_c,Reg)$. 
At each iteration of the algorithm, an NE of $G$ is computed by linear programming (note that the game considered here is not the whole game but a restricted one). Execution continues until convergence is detected. Convergence is achieved when the best-response oracles generate pure strategies that are already present in sets $S_x$ and $S_c$. In other words, convergence is obtained if for the current NE both players cannot improve their strategies by looking outside of the restricted game. One can also notice that at each step, due to Proposition \ref{pr:bound}, the regret obtained by $x = BR_x(P_\mathcal{U})$ against the mixed strategy $P_\mathcal{U}$ is always a lower bound on $OPT$. Therefore we can hope that a good lower bound on $OPT$ can be obtained before convergence and after only a few iterations of the algorithm.  

\begin{algorithm}[]
\DontPrintSemicolon
\KwData{Feasible solutions $\mathcal{X}$ and possible scenarios $\mathcal{U}$, 
singletons $S_x=\{x\}$ including an arbitrary feasible solution and $S_c=\{c\}$ including an arbitrary extreme scenario}
\KwResult{a (possibly mixed) NE}
converge $\leftarrow$ False\\
\While{converge is False}{
Find Nash equilibrium $(P_\mathcal{X},P_\mathcal{U}) \in G = (S_x,S_c,Reg)$\\
Find $x = BR_x(P_\mathcal{U})$ and $c = BR_c(P_\mathcal{X})$ \\
\lIf{$x\in S_x$ and $c\in S_c$}{converge $\leftarrow$ True}
\lElse 
{add $x$ to $S_x$ and $c$ to $S_c$}}
\Return $(P_\mathcal{X},P_\mathcal{U})$  

\caption{Double Oracle Algorithm}
\label{alg:doa1}
\end{algorithm}

The correctness of best-response-based double oracle algorithms for finite two-player zero-sum games has been established by McMahan et al \cite{McMahanGB03}; 
the intuition for this correctness is as follows. Once the algorithm converges, the current solution must be an equilibrium of the game, because each player's current strategy
is a best response to the other player's current strategy. This stems from the fact that the best-response oracle, which searches over all possible strategies, cannot find anything better. Furthermore, the algorithm must converge, because at worst, it will generate all pure strategies. 

Note that it is guaranteed that feasible solutions that are not \emph{weak} will never be generated by the double oracle algorithm. A feasible solution is said to be weak if it is optimal for some scenario. Subsequently, an element $i\in \{1,\ldots,n\}$ is said to be a weak element if it is part of a weak solution. Because each oracle solves a standard optimization problem for a specific scenario, a feasible solution that is not  weak will never be generated. A preprocessing step could therefore be done by removing non-weak elements. Note however that the difficulty of this preprocessing step is highly dependent on the problem under study. For instance, while deciding whether an edge is weak or not can be done in polynomial time for the robust spanning tree problem \cite{Yaman99minimumspanning}, it has been shown to be NP-complete for the robust shortest path problem \cite{KPY01}.

According to the problem instance and the uncertainty set over costs $\mathcal{U}$, the set of generated strategies for both players may either 1) remain relatively small or on the contrary 2) become very large. Therefore, the running time of the double oracle algorithm will be highly dependent on the fact that the problem instance belongs to case 1) or 2). Our intuition is that, for real world problems involving a small number of feasible solutions that deserve to be considered, the problem instance will yield case 1).
Secondly, for instances yielding the second situation we expect the double oracle algorithm to provide a good lower bound on $OPT$ after only a few iterations of the algorithm.

\section{Adapting the lower bound for a branch and bound procedure}
\label{sec:bb}
We now show how our lower bound can be adapted to be used in a branch and bound procedure for determining a minmax regret (pure) solution.
It is well-known that two key ingredients for a successful branch and bound procedure are: i) the accuracy of the lower bound, ii) the efficiency of its computation \cite{conforti2014integer}. In order to speed up the computation of the lower bound at each node, one can take advantage of the information obtained at the father node. This is the topic of this section.

In the branch and bound we use in the sequel of the paper, a node, defining a restricted set $\mathcal{X}' \subseteq \mathcal{X}$ of feasible solutions, is characterized by a couple $(IN(\mathcal{X'})$, $OUT(\mathcal{X'}))$, where $IN(\mathcal{X}')\subseteq \{i \in \{1,\ldots,n\}: x_i = 1,  \forall x \in \mathcal{X}'\}$ is the set of all elements that are enforced to be part of every feasible solution in $\mathcal{X}'$, and $OUT(\mathcal{X}') \subseteq \{i \in \{1,\ldots,n\}: x_i = 0, \forall x \in \mathcal{X}'\}$ is the set of all elements that are forbidden in any feasible solution of $\mathcal{X}'$. The branching scheme consists in partitioning $\mathcal{X}'$ into two sets by making mandatory or forbidden an element $k \in \{1,\ldots,n\}\setminus (IN(\mathcal{X'})\cup OUT(\mathcal{X'}))$:
\begin{itemize}
\item the first child is the set $\mathcal{X}'' \subseteq \mathcal{X}'$ characterized by $IN(\mathcal{X}'')=IN(\mathcal{X}')\cup \{k\}$ and $OUT(\mathcal{X}'')=OUT(\mathcal{X}')$,
\item the second child is the set $\mathcal{X}''' \subseteq \mathcal{X}'$ characterized by $IN(\mathcal{X}''')=IN(\mathcal{X}')$ and $OUT(\mathcal{X}''')=OUT(\mathcal{X}')\cup \{k\}$.
\end{itemize}
At a node $\mathcal{X}'$ of the branch and bound tree, the computation of the lower bound amounts to determine:
\begin{equation}
    \underset{P_{\mathcal{U}} \in \Delta_{\mathcal{U}}}{\max} \underset{x\in \mathcal{X'}}{\min} Reg(x,P_{\mathcal{U}})
\end{equation}


Our double oracle approach will be unchanged except that given a probability distribution $P_\mathcal{U}$ over extreme scenarios, the best response procedure of the  $x$-player will now return the best possible response in $\mathcal{X}'$. The double oracle algorithm will hence generate a restricted game $G_{\mathcal{X}'}\!=\! (S_x,S_c,Reg)$ where all strategies in $S_x$ are in $\mathcal{X}'$.

Two things can be noted to speed up the computation of the lower bound at nodes $\mathcal{X}''$ and $\mathcal{X}'''$. To initialize the set of feasible solutions in $G_{\mathcal{X}''}$ and $G_{\mathcal{X}'''}$, we partition the set $S_x$ of generated solutions in $G_{\mathcal{X}'}$ into two sets $S_x^k = \{x \in S_x : x_k = 1\}$ and $S_x^{\overline{k}}=\{x\in S_x : x_k = 0\}$. The former is used as initial set in $G_{\mathcal{X}''}$ and the latter in $G_{\mathcal{X}'''}$.
The handling of set $S_c$ differs. The set of generated scenarios is indeed not dependent on the explored search node as the branching scheme does not add constraints on the scenarios. Therefore at each new node the current set $S_c$ of all scenarios generated so far in the course of the branch and bound is the new initial set of scenarios in Algorithm \ref{alg:doa1}. 

\section{Application to the robust shortest path problem}

In this section we illustrate how our work can be specified for the robust shortest path problem \cite{Montemanni:2004:EAR:1007797.1007804,DBLP:journals/4or/MontemanniG05,DBLP:journals/orl/MontemanniGD04,KPY01}, regarding both the computation of $LB^*$ and the determination of a minmax regret path. For comparison, we describe some other approaches in the literature that have been used for these problems (computation of $LB^*$ and determination of a minmax regret path).

\subsection{Problem description}
In the robust shortest path problem, a directed graph $\mathcal{G}=(\mathcal{V},\mathcal{E})$ is given where $\mathcal{V}$ is a set of vertices, numbered from 1 to $|\mathcal{V}|$, and $\mathcal{E}$ is a set of edges. A starting vertex $s \in \mathcal{V}$ and a destination vertex $t\in\mathcal{V}$ are given and an interval  $[\underline{c}_{ij},\overline{c}_{ij}] \subset \mathbb{R}^+$ is associated to each edge $(i,j)\in \mathcal{E}$, where $\underline{c}_{ij}$ (resp. $\overline{c}_{ij}$) is a lower (resp. upper) bound on the cost induced by edge $(i,j)$. The set $\mathcal{U}$ of possible scenarios is then defined as the Cartesian product of these intervals: $\mathcal{U} = \times_{(i,j) \in \mathcal{E}} [\underline{c}_{ij},\overline{c}_{ij}]$. Without loss of generality, we assume that $s=1$ and $t=|\mathcal{V}|$. The set of all paths from $s$ to $t$ is denoted by $\mathcal{X}$. An example of a graph with interval data is shown in Figure~\ref{fig:intervalgraph}. In the robust shortest path problem we consider here, one wishes to determine a path $x \in \mathcal{X}$ that minimizes $\max_{c \in \mathcal{U}} val(x,c) - val^*(c)$, where $val(x,c)$ is the value of path $x$ in scenario $c$ and $val^*(c)$ is the value of a shortest path in scenario $c$.
\begin{figure}[!h] 
   \centering
    \scalebox{.9}{\begin{tikzpicture}[->,>=stealth',shorten >=1pt,auto,node distance=2.7cm,
                    semithick]
  \tikzstyle{every state}=[circle,draw,text=black]

  \node[state] (B)                    {$1$};
  \node[state] (C) [above right of=B] {$2$};
  \node[state] (D) [below right of=B] {$3$};
  \node[state] (E) [right of=C] {$4$};
  \node[state] (F) [right of=D] {$5$};
  \node[state] (T) [below right of=E] {$6$};

  \path (B) edge node {$[2,4]$} (C)
        (B) edge node {$[3,5]$} (D)
        (C) edge node {$[1,2]$} (D)
        (C) edge node {$[1,4]$} (E)
        (D) edge [right] node {$[2,3]$} (E)
	(D) edge node {$[2,3]$} (F)
	(E) edge node {$[2,3]$} (T)
        (F) edge node {$[1,2]$} (T)
         ;
\end{tikzpicture}}
    \caption{An example of a graph with interval data.}
    \label{fig:intervalgraph}
\end{figure}
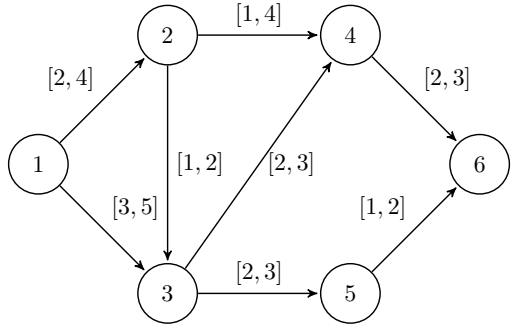

\subsection{Various approaches for computing $LB^*$}

In the experiments, we will compare the running times of our double oracle method for computing $LB^*$ with the following alternative approaches.

\subsubsection{LP formulation by Karasan et al.} 

The Robust Shortest Path problem (RSP) in the interval model was first studied by Karasan et al. \cite{KPY01} who provided the following mixed integer linear program $\mathcal{P}_{KPY}$ (from now on, we adopt the convention to use the initials of the authors as subscripts):   
\begin{empheq}[left=\mathcal{P}_{KPY}\empheqlbrace]{align}
\min& \sum_{(i,j)\in \mathcal{E}} \overline{c}_{ij}x_{ij}-\pi_{|\mathcal{V}|} \notag\\
\text{subject to}& \notag \\  
\pi_j \leq & \text{ } \pi_i + \underline{c}_{ij} + (\overline{c}_{ij} - \underline{c}_{ij})x_{ij} \text{, }\forall(i, j) \in \mathcal{E} \notag \\
b_j =&- \sum_{(i,j)\in\mathcal{E}}x_{ij} + \sum_{(j,k)\in\mathcal{E}}x_{jk} \notag \\
\pi_1 =& \text{ } 0 \notag \\
\pi_j \geq&\text{ }  0\text{, } j\in \{1,2,\ldots,|\mathcal{V}|\} \notag \\
x_{ij} \in& \{0,1\}& \notag 
\end{empheq}
where $b_j = 1$ if $j=1$, $b_j = -1$ if $j=|\mathcal{V}|$, and $b_j  = 0$ otherwise. Variable $x_{ij}$ indicates if edge $(i,j)$ belongs to the minmax regret solution and variable $\pi_i$ denotes the distance between $s$ and $i$ under the worst possible scenario regarding the solution path characterized by variables $x_{ij}$. We denote by $\widehat{\mathcal{P}}_{KPY}$ the relaxed version of this mixed integer linear program where $x_{ij} \ge 0$. One can notice that this relaxed version solves exactly the same problem as our double oracle algorithm:

\begin{prop}
Let $\widehat{LB}$ denote the optimal value of $\widehat{\mathcal{P}}_{KPY}$, then $\widehat{LB} = LB^*$.
\end{prop}

\begin{proof}
Note that $\widehat{\mathcal{P}}_{KPY}$ turns out to be a network flow problem, where $x_{ij}$ is the flow on edge $(i,j)$.

We first prove $\widehat{LB} \le LB^*$. Consider the mixed strategy of the $x$-player (that plays over paths) which realizes regret $LB^*$. Let us denote by $p_1,\ldots,p_q$ the corresponding probability distribution over paths $P_1,\ldots,P_q$ in $\mathcal{X}$. One can define a flow $x$ on $\mathcal{G}$ by setting $x_{ij} = \sum_{P_k:(i,j) \in P_k} p_k$. This flow is feasible for $\widehat{\mathcal{P}}_{KPY}$. The value of the objective function for this flow corresponds then exactly to $LB^*$, because:
\begin{align*}
\sum_{(i,j)\in \mathcal{E}} \overline{c}_{ij}x_{ij}-\pi_{|\mathcal{V}|} =& \max_{y\in\mathcal{X}} \sum_{(i,j) \in \mathcal{E}} \overline{c}_{ij}x_{ij} - \sum_{(i,j)\in\mathcal{E}}(\underline{c}_{ij} + (\overline{c}_{ij} - \underline{c}_{ij})x_{ij})y_{ij}\\
=&\max_{y\in\mathcal{X}} \sum_{(i,j) \in\mathcal{E}} x_{ij}(\overline{c}_{ij}(1-y_{ij}) + \underline{c}_{ij}y_{ij}) - \sum_{(i,j)\in\mathcal{E}}\underline{c}_{ij}y_{ij}\\
=&\max_{y\in\mathcal{X}} \sum_{(i,j) \in\mathcal{E}}(x_{ij} - y_{ij})\neg c_{ij}^y=\max_{y\in\mathcal{X}} \sum_{k=1}^q p_k Reg(P_k,\neg c^y)
\end{align*}
where $Reg(P_k,\neg c^y)$ is the regret induced by path $P_k$ in scenario $\neg c^y$. 

We now prove $LB^* \le \widehat{LB}$. The \emph{flow decomposition theorem} \cite{ahuja1993network} states that any feasible flow $x$ on graph $\mathcal{G}$ can be decomposed into at most $|\mathcal{E}|$ paths in $\mathcal{X}$ and cycles. In our case, it is easy to see that there always exists an optimal flow $x^*$ for $\widehat{\mathcal{P}}_{KPY}$ that does not include cycles. Such an optimal flow can therefore be decomposed into at most $|\mathcal{E}|$ paths in $\mathcal{X}$, denoted by $P_1,\ldots,P_q$ in the following. Let us consider the mixed strategy that plays each path $P_k$ ($k\in\{1,\ldots,q\}$) with a probability $p_k$ equal to the flow on the path. The regret of this mixed strategy is exactly $\widehat{LB}$, which can be shown by reversing the above sequence of equalities.
\end{proof}

Note that this connection between the game theoretic view and the relaxed versions of Mixed Integer Linear Programming (MILP) formulations of minimax regret problems apply to other robust combinatorial optimization problems. For instance, this connection is also valid for the MILP developed by Yaman et al. \cite{Yaman99minimumspanning} to solve the minmax regret spanning tree problem. Actually, Mastin \emph{et al.} proved that, for any robust combinatorial optimization problem, the value of a minmax regret mixed strategy corresponds to the optimal value of linear program $\mathcal{P}_{RO}'$ given at the end of Section~\ref{sec:game} (that may involve an exponential number of constraints) \cite{mastin2015randomized}. 

\subsubsection{LP formulation by Mastin et al. and a variant}
 
We present below the linear program $\mathcal{P}_{MJC}$ obtained by specifying $\mathcal{P}_{RO}'$ to the case of RSP\footnote{Note that the obtained linear program is very close to the Benders reformulation of $\mathcal{P}_{KPY}$ proposed by Montemanni and Gambardella \cite{DBLP:journals/4or/MontemanniG05}.}. Note that this program returns for each edge $(i,j)$ a probability $x_{ij}$, from which a probability distribution over paths in $\mathcal{X}$ can be inferred in polynomial time.
\begin{empheq}[left=\mathcal{P}_{MJC}\empheqlbrace]{align}
& \min_{\pi,x_{ij}:(i,j)\in\mathcal{E}} \sum_{(i,j)\in\mathcal{E}} \overline{c}_{ij} x_{ij} - \pi \notag \\
\text{subject to}& \notag \\
\pi \leq& \sum_{(i,j)\in\mathcal{E}} \!\!y_{ij}(\underline{c}_{ij} \!+\! x_{ij}(\overline{c}_{ij}\!-\!\underline{c}_{ij}))    \qquad  \forall y\! \in\! \mathcal{X}\label{c:regSP} \\
b_j =&- \sum_{(i,j)\in\mathcal{E}}x_{ij} + \sum_{(j,k)\in\mathcal{E}}x_{jk} \label{c:flow} \\
& x_{ij} \ge 0  \qquad  \forall (i,j) \in \mathcal{E} \notag
\end{empheq}
\normalsize

Constraints~(\ref{c:regSP}) are the specification of constraints~(\ref{eq:cstexp}) in the case of RSP, and therefore insure, at optimum, that $\sum_{(i,j)\in\mathcal{E}} \overline{c}_{ij} x_{ij} - \pi$ is the value of a min max regret mixed strategy. Furthermore, flow constraints~(\ref{c:flow}) make sure that variables $x_{ij}$ are consistent, i.e. they induce a probability distribution over paths in $\mathcal{X}$. It is thus possible to compute $LB^*$ by solving this linear program with a cutting-plane algorithm. 

In the experiments we also tested the linear program $\mathcal{D}_{MJC}$ where one adopts the dual viewpoint (viewpoint of the $c$-player that aims at maximizing the regret)\footnote{Even if it is a dual viewpoint, note that $\mathcal{D}_{MJC}$ is not the dual program of $\mathcal{P}_{MJC}$.}:

\begin{empheq}[left=\mathcal{D}_{MJC}\empheqlbrace]{align}
&\max_{\mu,y_{ij}:(i,j)\in\mathcal{E}} ~\mu - \sum_{(i,j)\in\mathcal{E}} \underline{c}_{ij} y_{ij} \notag \\
\text{subject to}& \notag \\ 
\mu \leq& \sum_{(i,j)\in\mathcal{E}} \!\!x_{ij}(\underline{c}_{ij}y_{ij} \!+\! (1-y_{ij})\overline{c}_{ij})    \qquad  \forall x\! \in\! \mathcal{X}\label{c:regSc} \\
b_j =&- \sum_{(i,j)\in\mathcal{E}}y_{ij} + \sum_{(j,k)\in\mathcal{E}}y_{jk} \notag \\
& y_{ij} \ge 0 \qquad  \forall (i,j) \in \mathcal{E}  \notag
\end{empheq}
Flow variables $y_{ij}$ induce a mixed strategy $P_{\mathcal{X}}$ over paths: the flow can indeed be decomposed into paths by the flow decomposition theorem (each cycle could be removed), and each path $z \in \mathcal{X}$ among them has a probability $P_{\mathcal{X}}(z)$ to be picked equal to the flow on $z$.
Constraints~(\ref{c:regSc}) ensure that $\mu$ is the minimum expected cost for the $x$-player given the mixed strategy of the $c$-player defined by variables $y_{ij}$ because:
\[
Reg(x,P_{\mathcal{X}}) = \sum_{(i,j)\in\mathcal{E}} \!\!\left(x_{ij}(\underline{c}_{ij}y_{ij} \!+\! (1-y_{ij})\overline{c}_{ij})\right) - \sum_{(i,j)\in\mathcal{E}} \underline{c}_{ij} y_{ij}
\]
for all $x \in \mathcal{X}$ and $y \in CH(\mathcal{X})$. Let us now explain this equation. After committing to $P_{\mathcal{X}}$, the $c$-player will draw a path including edge $(i,j)$ with probability $y_{ij}$. If the drawn path contains $(i,j)$ then cost $c_{ij}$ is set to $\underline{c}_{ij}$. Thus, the expected cost of edge $(i,j)$ for the $c$-player is $\underline{c}_{ij}y_{ij}$. If the drawn path does not contain $(i,j)$ then cost $c_{ij}$ is set to $\overline{c}_{ij}$. Consequently, if the $x$-player chooses path $x$, the expected cost of edge $(i,j)$ is $x_{ij}(\underline{c}_{ij}y_{ij} \!+\! (1-y_{ij})\overline{c}_{ij})$. The equation then follows from linearity of expectation w.r.t. edge costs.

This analysis shows that $\mathcal{D}_{MJC}$ computes $\max_{P_\mathcal{X} \in \Delta_\mathcal{X}} \min_{x \in \mathcal{X}} Reg(x,P_\mathcal{X})$ where $P_\mathcal{X}$ induces the mixed strategy $P_\mathcal{U} \in \Delta_\mathcal{U}$ defined by $P_\mathcal{U}(\neg c^y) = P_\mathcal{X}(y)$. We recall that $\mathcal{P}_{MJC}$ computes $\min_{P_{\mathcal{X}} \in \Delta_\mathcal{X}}  \max_{y \in \mathcal{X}} Reg(P_{\mathcal{X}},y)$ where $y$ induces the pure strategy $\neg c^y \in \mathcal{U}$. By the minimax theorem, the optimal values of $\mathcal{P}_{MJC}$ and $\mathcal{D}_{MJC}$ therefore coincide.


Note that there are as many constraints in $\mathcal{P}_{MJC}$ and $\mathcal{D}_{MJC}$ as there are paths in $\mathcal{X}$. We now detail how to overcome this difficulty with a cutting plane method. Given a subset $\mathcal{X}' \subseteq \mathcal{X}$ of paths, let us denote by $\mathcal{P}_{MJC}(\mathcal{X}')$ (resp. $\mathcal{D}_{MJC}(\mathcal{X}')$) the linear program where one considers only constraints induced by paths in $\mathcal{X}'$. Starting from a small subset $\mathcal{X}'$ of paths, one solves $\mathcal{P}_{MJC}(\mathcal{X}')$ (resp. $\mathcal{D}_{MJC}(\mathcal{X}')$). There are two possibilities: either the found optimal solution is feasible for $\mathcal{P}_{MJC}$ (resp. $\mathcal{D}_{MJC}$), in which case it is optimal for $\mathcal{P}_{MJC}$ (resp. $\mathcal{D}_{MJC}$), or it violates at least one constraint in $\mathcal{P}_{MJC}$ (resp. $\mathcal{D}_{MJC}$).
To identify the current situation, we solve the shortest path problem with cost function $c_{ij} = \underline{c}_{ij} +  x_{ij}(\overline{c}_{ij} - \underline{c}_{ij}) $ (resp.  $c_{ij} = \underline{c}_{ij}y_{ij} \!+\! (1-y_{ij})\overline{c}_{ij}$). 
If the cost of that path is strictly smaller than $\pi$ (resp. $\mu$), then it means that at least one constraint is violated, and any constraint associated to a shortest path is a most violated one. In this case, we add to $\mathcal{X}'$ the obtained shortest path and we reoptimize. If the regret induced by a shortest path is greater than or equal to $\pi$ (resp. $\mu$), then it means that no constraint is violated in $\mathcal{P}_{MJC}$ (resp. $\mathcal{D}_{MJC}$), and therefore we have solved $\mathcal{P}_{MJC}$ (resp. $\mathcal{D}_{MJC}$).

\subsubsection{Our approach}

In order to specify Algorithm \ref{alg:doa1} for the robust shortest path problem, one needs to specify the two oracles $BR_x(P_{\mathcal{U}})$ (best response of the $x$-player to a mixed strategy $P_{\mathcal{U}}$ of the $c$-player) and $BR_c(P_{\mathcal{X}})$ (best response of the $c$-player to a mixed strategy $P_{\mathcal{X}}$ of the $x$-player). These two oracles can be implemented by using Dijkstra's algorithm. 
\begin{itemize}
\item For $BR_x(P_{\mathcal{U}})$ the cost of $(i,j)$ is defined as $\sum_{c \in \mathcal{U}} P_{\mathcal{U}}(c) c_{ij}$. Any optimal path found by running Dijkstra's algorithm is then a best response (by Observation~\ref{obs:xOracle}).  
\item For $BR_c(P_{\mathcal{X}})$ the cost of $(i,j)$ is defined as $\underline{c}_{ij} + \sum_{x \in \mathcal{X}}P_\mathcal{X}(x)x_{ij}(\overline{c}_{ij} - \underline{c}_{ij}) $. For any optimal path $z$ found by running Dijkstra's algorithm, $\neg c^{z}$ is then a best response (by Proposition~\ref{propBestCResponse}).  
\end{itemize}

As the robust shortest path problem is defined as a single-source single-target problem, we implemented a bidirectional variant of Dijkstra's algorithm \cite{pohl1969bi} that make it possible to find an optimal path by only exploring a small part of the graph. Consequently, the double oracle algorithm has the advantage over the linear programming formulations that it does not need an explicit representation of graph $\mathcal{G}$.


For this reason, and as it will be shown by the numerical tests, depending on the structure of graph $\mathcal{G}$, running the double oracle algorithm is much faster than solving $\widehat{\mathcal{P}}_{KPY}$ by linear programming. Regarding the comparison with the solution of $\mathcal{P}_{MJC}$ and $\mathcal{D}_{MJC}$, both the double oracle algorithm and the cutting plane algorithms make calls to the oracles defined above. These three methods are betting on a low number of ``relevant'' paths in the problem so that it could converge in few iterations. However, while the double oracle uses both oracles $BR_c(P_{\mathcal{X}})$ and $BR_c(P_{\mathcal{X}})$, the cutting plane algorithm for solving $\mathcal{P}_{MJC}$ (resp. $\mathcal{D}_{MJC}$) only uses $BR_c(P_{\mathcal{X}})$ (resp. $BR_x(P_{\mathcal{U}})$) and an implicit representation of \emph{all} strategies of the $x$-player (resp. $c$-player) through flow constraints on variables $x_{ij}$ (resp. $y_{ij}$), one variable per edge $(i,j)$ in $\mathcal{G}$. The sequence of LPs that must be solved can therefore be costly for large graphs, especially the first LP in the sequence (the other LPs are computed much faster thanks to reoptimization).

\subsection{Branch and bound algorithm for the minmax regret shortest path problem}

In the experiments, we will compare the performances of our branch and bound procedure for determining a minmax regret path with other branch and bounds proposed in the literature.

\subsubsection{Branch and bound by Montemanni et al.}

Montemanni et al.~\cite{DBLP:journals/orl/MontemanniGD04} proposed a branch and bound algorithm for the minmax regret shortest path problem. The branching strategy generates a search tree where each node corresponds to a subset $\mathcal{X}'$ of feasible solutions characterized by a subset $IN(\mathcal{X}')$ of mandatory edges and a subset $OUT(\mathcal{X}')$ of forbidden edges. Note that the following property is guaranteed by the branching strategy: $IN(\mathcal{X}')$ always contains a chain of connected edges which form a subpath starting from $s$. For a detailed presentation of the branching strategy, the reader is referred to the article by Montemanni et al. \cite{DBLP:journals/orl/MontemanniGD04}. The authors use the following lower bound $LB_{MGD}(\mathcal{X}')$ on the max regret of a subset $\mathcal{X}'$ of feasible solutions defined by $IN(\mathcal{X}')$ and $OUT(\mathcal{X}')$:
\small
\begin{equation}
LB_{MGD}(\mathcal{X}') = SP(c^{\mathcal{E}},IN(\mathcal{X}'),OUT(\mathcal{X}')) - SP(c^{\mathcal{E}\backslash OUT(\mathcal{X}')},\emptyset,\emptyset)
\end{equation}
\normalsize

where $SP(c, A, B)$ is the value of the shortest $s-t$ path including all edges from $A$ and excluding all edges from $B$ in scenario $c$, and $c^{A}$ is the scenario where all the edges in $A$ are set to their upper bound and all others are set to their lower bound. The computation of $LB_{MGD}(\mathcal{X}')$ requires two runs of Dijkstra's algorithm: one run for computing $SP(c^{\mathcal{E}},IN(\mathcal{X}'),OUT(\mathcal{X}'))$ and one run for computing $SP(c^{\mathcal{E}\backslash OUT(\mathcal{X}')},\emptyset,\emptyset)$.

\subsubsection{Branch and bound by Chassein and Goerigk}

Chassein and Goerigk \cite{DBLP:journals/eor/ChasseinG15} used the same branching strategy as Montemanni et al., but changed the lower bounding procedure as detailed in section~\ref{sec:secCG}. When focusing on a subset $\mathcal{X}'$ of feasible solutions, we denote by $LB_{CG}(\mathcal{X}')$ the lower bound adapted from $LB_{CG}$. 
In order to define $LB_{CG}(\mathcal{X}')$, the authors consider the set $\hat{\Delta}_{\mathcal{U}}(\mathcal{X})$ of mixed strategies with two equally likely extreme scenarios $c^1$ and $c^2$ such that $c^1_{ij} = c^2_{ij} = \overline{c}_{ij}$ (resp. $\underline{c}_{ij}$) if $(i,j) \in IN(\mathcal{X}')$ (resp. $OUT(\mathcal{X}')$) and $c^1_{ij} + c^2_{ij} = \underline{c}_{ij}+\overline{c}_{ij}$ if $(i,j) \not\in IN(\mathcal{X}')\cup OUT(\mathcal{X}')$. 
The lower bound $LB_{CG}(\mathcal{X}')$ then reads as follows:
\small
\begin{equation}
\max \left\{ \sum_{(i,j) \in IN(\mathcal{X}')}\!\!\!\!\!\! \overline{c}_{ij} - \hat{c}_{ij} + SP(\hat{c},IN(\mathcal{X}'),OUT(\mathcal{X}')) - \mathbb{E}_{P_\mathcal{U}}(val^*(C)) : P_\mathcal{U} \in \hat{\Delta}_{\mathcal{U}}(\mathcal{X})\right\}\notag
\end{equation}
\normalsize
where $\hat{c}_{ij} = (\underline{c}_{ij}+\overline{c}_{ij})/2$, $C$ is a random variable on $\mathcal{U}$, and $\mathbb{E}_{P}(X)$ denotes the expectancy of random variable $X$ given the probability distribution $P$ on the domain of $X$. The maximization operation is performed by resorting to Suurballe's algorithm \cite{journals/networks/SuurballeT84}. Overall, the computation of $LB_{CG}(\mathcal{X}')$ requires three runs of Dijkstra's algorithm: one run for computing $SP(\hat{c},IN(\mathcal{X}')$, $OUT(\mathcal{X}'))$ and two runs involved in Suurballe's algorithm.

\subsubsection{Our approach}

Our branch and bound also uses the branching strategy proposed by Montemanni et al., but relies on the double oracle algorithm for the lower bounding procedure. In the double oracle algorithm, the best responses are computed by resorting to Dijkstra's algorithm both for the $x$-player and the $c$-player:
\begin{itemize}
\item $x$-player: the best response procedure uses Dijkstra's algorithm to find the best $s_{in}-t$ path according to the scenario as defined in Observation~\ref{obs:xOracle}, where $s_{in}$ is the last node on the subpath defined by $IN(\mathcal{X}')$ and where all edges in $OUT(\mathcal{X}')$ have costs set to infinity. The subpath defined by $IN(\mathcal{X}')$ and this path are then concatenated to obtain the best $x$-response. 
\item $c$-player: the best response procedure uses Dijkstra's algorithm to find the best $s-t$ path according to the scenario defined in Proposition~\ref{prop:cOracle}. The best $c$-response is then obtained by setting the costs of edges along this path to their lower bound, and the costs of other edges to their upper bound.
\end{itemize}

As all the best responses already generated are transmitted from a search tree node to its children (see the end of Section \ref{sec:bb}), the computation of the lower bound speeds up with the depth of the considered node in the branch and bound procedure. 


\subsection{Experiments}
We evaluate our approach using two different experiments. In a first experiment, we compare the accuracy and computation time of different lower bounds on the minmax regret. In a second experiment, we test how a branch and bound procedure using the lower bound investigated in this paper compares with the previous results by Chassein and Goerigk \cite{DBLP:journals/eor/ChasseinG15} and Montemanni et al. \cite{DBLP:journals/orl/MontemanniGD04}. 

\textbf{Configuration and implementation details.} All times are wall-clock times on a 2.4GHz Intel Core i5 machine with 8GB main memory. Our implementation is in C++, with external calls to GUROBI version 5.6.3 when (possibly mixed) linear programs need to be solved. Lastly, all single-source single-target Dijkstra algorithms are implemented in a bidirectional fashion and the lower bound designed by Chassein and Goerigk is computed using Suurballe's algorithm \cite{journals/networks/SuurballeT84}. All results are averaged over 100 runs.

\textbf{Description of the instances.} Chassein and Goerigk \cite{DBLP:journals/eor/ChasseinG15} considered two different types of randomly generated graph classes. We consider the same graph
classes and respect the same experimental protocol.

The generation of the costs is parametrized by two parameters, namely $r$ and $d$. For an edge $(i,j)$, a random number $m$ is first sampled from the interval $[1, r]$. The lower bound cost $\underline{c}_{ij}$ is sampled uniformly from $[ ( 1 - d ) m, ( 1 + d ) m]$ and the upper bound cost $\overline{c}_{ij}$ is then chosen uniformly from interval $[\underline{c}_{ij}, ( 1 + d ) m]$. Thus, parameter $d$ enables to control the cost variability (the cost variability increases with $d$).

The graph family $R$-$n$-$r$-$d$-$\delta$ consists of randomly generated graphs such that each graph has $n$ nodes and an approximate edge density of
$\delta$ (i.e. the probability that an edge exists is $\delta$). The starting vertex $s$ is the first generated node of the graph and the destination vertex $t$ is the last generated node of the graph.

The second graph family $K$-$n$-$r$-$d$-$w$ consists of layered graphs. Every layer is completely connected to the next layer. The starting node $s$ is connected to the first layer and the last layer is connected to the destination node $t$. The overall graph consists of $n$ nodes where every layer contains $w$ nodes.

\textbf{First Experiment.} In this experiment we compare the computation time and accuracy of several lower bounds on the minmax regret value:
\begin{itemize}
  \item We will use $LB^*_n$ to denote the lower bound obtained when the double oracle algorithm is stopped after $n$ iterations regardless of the fact that convergence has been attained or not. The returned lower bound is then $Reg(BR_x(P_\mathcal{U}),P_\mathcal{U})$, where $BR_x(P_\mathcal{U})$ is the last best $x$-response generated. 
 \item The lower bound studied in this paper is denoted by $LB^*$. We may also use the notations $LB^*_{DO}$, $LB^*_{\widehat{\mathcal{P}}_{KPY}}$, $LB^*_{\mathcal{P}_{MJC}}$ or $LB^*_{\mathcal{D}_{MJC}}$  according to the method used to compute $LB^*$: 
 \begin{itemize}
 \item Notation $LB^*_{DO}$ means that $LB^*$ is computed with the double oracle algorithm ran until convergence.
 \item Notation $LB^*_{\widehat{\mathcal{P}}_{KPY}}$ means that $LB^*$ is computed using linear program $\widehat{\mathcal{P}}_{KPY}$.
 \item Notation $LB^*_{\mathcal{P}_{MJC}}$ means that $LB^*$ is computed using the cutting plane method to solve $\mathcal{P}_{MJC}$.
 \item Notation $LB^*_{\mathcal{D}_{MJC}}$ means that $LB^*$ is computed using the cutting plane method to solve $\mathcal{D}_{MJC}$.
 \end{itemize}
 
 \item The lower bound designed by Chassein and Goerigk is denoted by $LB_{CG}$. 
 \item We recall that, as shown by Kasperski and Zielinski \cite{KasperskiZ06}, the regret of the midpoint solution is not more than $2 \cdot OPT$. Thus, the regret of the midpoint solution divided by 2 is a lower bound on $OPT$, which we denote by $LB_{KZ}$. 
\end{itemize}

To compare the accuracy of these lower bounds, the following ratios are computed. The closer they are to 1, the more accurate are the lower bounds. 
\begin{itemize}
\item We denote by Gap-medSol the approximation guarantee obtained for the midpoint solution. It corresponds to the regret of the midpoint solution divided by the considered lower bound. 
\item When considering the double oracle algorithm, we will design by minSol the element of minimum regret in the set composed of the midpoint solution and of the admissible solutions of the restricted game (i.e., generated by the double oracle algorithm). Gap-minSol represents the obtained approximation guarantee for the minSol solution. It corresponds to the regret of the minSol solution divided by the considered lower bound. 
\item Lastly, we denote by Gap-Opt the gap between the considered lower bound and the optimal value $OPT$. It corresponds to $OPT$ divided by the considered lower bound. The value $OPT$ is computed using $\mathcal{P}_{KPY}$.
\end{itemize}

The results are presented in Tables 1 to 4 for R-graphs and 5 to 8 for K-graphs\footnote{We highlight the fact that  Gap-medSol for $LB_{KZ}$ is equal to 1 (and not 2) for instances where the regret of the midpoint solution is 0. Consequently the average Gap-medSol, somewhat counter-intuitively, can be strictly less than 2 in some tables.}. Notation ``inf'' stands for an infinite gap (obtained if the lower bound is 0 with $OPT\neq 0$). We can make the following observations.

For R-graphs, we see that the computation of $LB^*_{DO}$ is much faster than the computation of $LB^*_{\widehat{\mathcal{P}}_{KPY}}$ while it is the opposite in K-graphs. Indeed, intuitively, in R-graphs few paths will be interesting (paths with few edges) and the double oracle algorithm will converge after few iterations while in K-graphs, the number of interesting paths may explode with the size of the graph (all paths have the same number of edges).  

For the same reason, the computation of $LB^*_{\mathcal{P}_{MJC}}$ and $LB^*_{\mathcal{D}_{MJC}}$ are more effective than the one of $LB^*_{\widehat{\mathcal{P}}_{KPY}}$ for R-graphs and less effective for K-graphs. The computation times of $LB^*_{\mathcal{P}_{MJC}}$ and $LB^*_{\mathcal{D}_{MJC}}$ are similar on R-graphs but the computation of  $LB^*_{\mathcal{P}_{MJC}}$ is much more efficient on K-graphs than the one of $LB^*_{\mathcal{D}_{MJC}}$.  While the computation of $LB^*_{\mathcal{P}_{MJC}}$ and $LB^*_{\mathcal{D}_{MJC}}$ are more effective than the computation of $LB^*_{DO}$ on K-graphs, they are clearly outperformed on R-graphs. For those graphs the double oracle algorithm is able to focus on a small portion of the graph while the restricted linear programs in the cutting plane method must take into account the whole graph. 

For both types of graphs, $LB^*$ is of course more computationally demanding than $LB_{KZ}$ and $LB_{CG}$ but it is much more accurate and this is often true after only a few iterations of the double oracle algorithm (see $LB^*_{15}$ for instance). 

Finally, note that Gap-minSol can be slightly better than Gap-medSol, but globally, it does not improve much the approximation guarantee. The reason is that the regret of the midpoint solution is often close to $OPT$, as confirmed by the numerical data.\\

\begin{table*}[p]
\footnotesize
\resizebox{\textwidth}{!}{
\begin{tabular}{| l || l | r | r | r || r | r | r | r || r | r | r | r || r | r | r | r |}
  \hline
   & \multicolumn{4}{|c|}{time for lb (ms)} 
   & \multicolumn{4}{|c|}{Gap-medSol}
   & \multicolumn{4}{|c|}{Gap-minSol}
   & \multicolumn{4}{|c|}{Gap-Opt}\\
  \hline
   & mean & std & min & max & mean & std & min & max & mean & std & min & max & mean & std & min & max \\
  \hline
  MILP & 2.95 & 3.08 & 1 &23 &-&-&-&-& -&-&-&-& -&-&-&-\\
  \hline
  $LB^*_5$ & 2.33 & 1.75 & 0 & 10 & 1.15 & 0.22 & 1.00 & 1.88 & 1.14 & 0.22 &1.00 & 1.88 & 1.14 & 0.22 & 1.00 & 1.88 \\
  \hline
  $LB^*_{10}$ & 2.57 & 1.92 & 0 & 9 & 1.15 & 0.22 & 1.00 & 1.88 & 1.14 & 0.22 & 1.00 & 1.88 & 1.14 & 0.22 & 1.00 & 1.88 \\
  \hline  
  $LB^*_{15}$ & 2.60 & 1.89 & 0 & 7 & 1.15 & 0.22 & 1.00 & 1.88 & 1.14 & 0.22 & 1.00 & 1.88 & 1.14 & 0.22 & 1.00 & 1.88\\
  \hline
 $LB^*_{20}$ & 2.43 & 1.83 & 0 & 9 & 1.15 & 0.22 & 1.00 & 1.88 & 1.14 & 0.22 & 1.00 & 1.88 & 1.14 & 0.22 & 1.00 & 1.88\\
  \hline
  $LB^*_{DO}$  & 2.51 & 1.95 & 0 & 9 &1.15 & 0.22 & 1.00 & 1.88 & 1.14 & 0.22 & 1.00 & 1.88 & 1.14 & 0.22 & 1.00 & 1.88\\
\hline
   $LB^*_{\widehat{\mathcal{P}}_{KPY}}$ & 0.61 & 0.49 & 0 & 1 & 1.15 & 0.22 & 1.00 & 1.88 & - & - & - & - & 1.14 & 0.22 & 1.00 & 1.88\\ 
  \hline
  $LB^*_{\mathcal{P}_{MJC}}$ &   0.65 & 0.59 & 0 & 3 & 1.15 & 0.22 & 1.00 & 1.88 &-&-&-&-& 1.14 & 0.22 & 1.00 & 1.88\\
  \hline
   $LB^*_{\mathcal{D}_{MJC}}$& 0.56 & 0.52 & 0 & 2 & 1.15 & 0.22 & 1.00 & 1.88 &-&-&-&-& 1.14 & 0.22 & 1.00 & 1.88\\
   \hline
  $LB_{CG}$ & 0.08 & 0.27 & 0 & 1 & 1.45 & 0.49 & 1.00 & 2.00 & -  & -  & -  & -  & 1.45 & 0.49 & 1.00 & 2.00 \\
    \hline
  $LB_{KZ}$ & 0.06 & 0.24 & 0 & 1 & 1.47 & 0.50 & 1.00 & 2.00 & - & - & - & - & 1.46 & 0.50 & 1.00 & 2.00 \\
 \hline
\end{tabular}}

\caption{R10-1000-0.5-1}
\end{table*}

\begin{table*}
\footnotesize
\resizebox{\textwidth}{!}{
\begin{tabular}{| l || r | r | r | r || r | r | r | r || r | r | r | r || r | r | r | r |}
  \hline
   & \multicolumn{4}{|c|}{time for lb (ms)} 
   & \multicolumn{4}{|c|}{Gap-medSol}
   & \multicolumn{4}{|c|}{Gap-minSol}
   & \multicolumn{4}{|c|}{Gap-Opt}\\
  \hline
   & mean & std & min & max & mean & std & min & max & mean & std & min & max & mean & std & min & max\\
 \hline
  MILP & 96.36 & 21.28 & 63 & 155 & - &- &-&-& -&-&-&-& -&-&-&-\\
  \hline
  $LB^*_5$ & 6.28 & 3.95 & 1 & 17 & 1.25 & 0.51 & 1.00 & 4.61 & 1.24 & 0.51 & 1.00 & 4.61 & 1.24 & 0.50 & 1.00 & 4.61\\
  \hline
  $LB^*_{10}$ & 6.50 & 4.55 & 1 & 24 & 1.18 & 0.24 & 1.00 & 1.99 & 1.18 & 0.24 & 1.00 & 1.99 & 1.18 & 0.24 & 1.00 & 1.99\\
  \hline  
  $LB^*_{15}$ & 6.60 & 4.64 & 1 & 28 & 1.18 & 0.24 & 1.00 & 1.99 & 1.18 & 0.24 & 1.00 & 1.99 & 1.18 & 0.24 & 1.00 & 1.99\\
  \hline
  $LB^*_{20}$ & 6.37 & 4.34 & 1 & 26 & 1.18 & 0.24 & 1.00 & 1.99 & 1.18 & 0.24 & 1.00 & 1.99 & 1.18 & 0.24 & 1.00 & 1.99\\
  \hline
  $LB^*_{DO}$  & 6.14 & 4.10 & 1 & 25 & 1.18 & 0.24 & 1.00 & 1.99 & 1.18 & 0.24 & 1.00 & 1.99 & 1.18 & 0.24 & 1.00 & 1.99\\
   \hline 
  $LB^*_{\widehat{\mathcal{P}}_{KPY}}$ & 17.64 & 2.66 & 13 & 25 & 1.18 & 0.24 & 1.00 & 1.99 & - & - & - & - & 1.18 & 0.24 & 1.00 & 1.99\\
  \hline
  $LB^*_{\mathcal{P}_{MJC}}$ & 10.38 & 4.34 & 4 & 24 & 1.18 & 0.24 & 1.00 & 1.99 & - & - & - & - & 1.18 & 0.24 & 1.00 & 1.99\\
  \hline
  $LB^*_{\mathcal{D}_{MJC}}$ & 8.51 & 3.08 & 4 & 25 & 1.18 & 0.24 & 1.00 & 1.99 & - & - & - & - & 1.18 & 0.24 & 1.00 & 1.99\\
  \hline
  $LB_{CG}$ & 1.51 & 0.62 & 0 & 3 & 1.64 & 0.47 & 1.00 & 2.00 & - & - & - & - & 1.63 & 0.46 & 1.00 & 2.00 \\
  \hline
  $LB_{KZ}$ & 0.89 & 0.49 & 0 & 2 & 1.67 & 0.47 & 1.00 & 2.00 & -&-&-&-& 1.67 & 0.47 & 1.00 & 2.00\\
  \hline
\end{tabular}}
\caption{R100-1000-0.5-0.5}
\end{table*}

\begin{table*}
\footnotesize
\resizebox{\textwidth}{!}{
\begin{tabular}{| l || r | r | r | r || r | r | r | r || r | r | r | r || r | r | r | r |}
  \hline
   & \multicolumn{4}{|c|}{time for lb (ms)} 
   & \multicolumn{4}{|c|}{Gap-medSol}
   & \multicolumn{4}{|c|}{Gap-minSol}
   & \multicolumn{4}{|c|}{Gap-Opt} \\
  \hline
   & mean & std & min & max & mean & std & min & max & mean & std & min & max & mean & std & min & max \\
  \hline
  MILP& 8729.55 & 1113.63 & 5967 & 11226 &-&-&-&-&-&-&-&-&-&-&-&-\\
  \hline
$LB^*_5$ & 47.29 & 28.43 & 3 & 96 & 1.30 & 0.41 & 1.00 & 4.39 & 1.28 & 0.41 & 1.00 & 4.39 & 1.28 & 0.40 & 1.00 & 4.39\\
  \hline
$LB^*_{10}$ & 53.67 & 36.84 & 4 & 127 & 1.23 & 0.24 & 1.00 & 1.96 & 1.21 & 0.23 & 1.00 & 1.96 & 1.21 & 0.23 & 1.00 & 1.96\\
  \hline  
$LB^*_{15}$& 53.25 & 36.74 & 3 & 127 & 1.23 & 0.24 & 1.00 & 1.96 & 1.21 & 0.23 & 1.00 & 1.96 & 1.21 & 0.23 & 1.00 & 1.96\\
  \hline
$LB^*_{20}$ & 52.72  & 36.01 & 3 & 128 &  1.23 & 0.24 & 1.00 & 1.96 & 1.21 & 0.23 & 1.00 & 1.96 & 1.21 & 0.23 & 1.00 & 1.96\\
  \hline
$LB^*_{DO}$  & 54.98 & 37.47 & 4 & 127 &1.23 & 0.21 & 1.00 & 1.96 & 1.21 & 0.23 & 1.00 & 1.96 & 1.21 & 0.23 & 1.00 & 1.96\\
  \hline 
$LB^*_{\widehat{\mathcal{P}}_{KPY}}$ & 1734.40 & 218.08 &  847 & 2166 & 1.23 & 0.21 & 1.00 & 1.96 & - & - & - & - & 1.21 & 0.23 & 1.00 & 1.96\\
  \hline
$LB^*_{\mathcal{P}_{MJC}}$ & 501.34 & 72.44 & 378 & 650 & 1.23 & 0.21 & 1.00 & 1.96 & - & - & - & - & 1.21 & 0.23 & 1.00 & 1.96\\
  \hline
$LB^*_{\mathcal{D}_{MJC}}$ & 542.17 & 196.21 &  283 & 1410 & 1.23 & 0.21 & 1.00 & 1.96 & - & - & - & - & 1.21 & 0.23 & 1.00 & 1.96\\
  \hline
  $LB_{CG}$& 22.32 & 2.97 & 16 & 31 & 1.73 & 0.44 & 1.00 & 2& - & - & - & - & 1.71 & 0.43 & 1.00  & 2.00 \\
  \hline
  $LB_{KZ}$ & 8.10 & 2.23 & 2 & 16 & 1.74 & 0.44 & 1.00 & 2.00 & -&-&-&-& 1.72 & 0.44 & 1.00 & 2.00 \\
\hline
\end{tabular}}
\caption{R500-1000-0.5-0.5}
\end{table*}

\begin{table*}
\footnotesize
\resizebox{\textwidth}{!}{
\begin{tabular}{| l || r | r | r | r || r | r | r | r || r | r | r | r || r | r | r | r |}
  \hline
   & \multicolumn{4}{|c|}{time for lb (ms)} 
   & \multicolumn{4}{|c|}{Gap-medSol}
   & \multicolumn{4}{|c|}{Gap-minSol}
   & \multicolumn{4}{|c|}{Gap-Opt} \\
  \hline
   & mean & std & min & max & mean & std & min & max & mean & std & min & max & mean & std & min & max \\
  \hline
  MILP & 59920.20 & 7593.73 & 43871 & 83568&-&-&-&-&-&-&-&-&-&-&-&-\\
  \hline
  $LB^*_5$ &131.93 & 73.71 & 8 & 285 & 1.31 & 0.37 & 1.00 & 3.47 & 1.29 & 0.33 & 1.00 & 3.19 & 1.29 & 0.33 & 1.00 & 3.19 \\
  \hline
  $LB^*_{10}$ & 145.32 & 97.80 & 7 & 454 &1.28 & 0.28 & 1.00 & 1.97 & 1.26 & 0.27 & 1.00 & 1.97 & 1.26 & 0.26 & 1.00 & 1.97 \\
  \hline  
  $LB^*_{15}$ & 145.44 & 97.22 & 8 & 460 & 1.28 & 0.28 & 1.00 & 1.97 & 1.26 & 0.27 & 1.00 & 1.97 & 1.26 & 0.26 & 1.00 & 1.97\\
  \hline
  $LB^*_{20}$ & 146.04 & 98.59 & 8 & 455 & 1.28 & 0.28 & 1.00 & 1.97 & 1.26 & 0.27 & 1.00 & 1.97 & 1.26 & 0.26 & 1.00 & 1.97 \\
  \hline
  $LB^*_{DO}$  & 146.47 & 98.45 & 8 & 458 &  1.28 & 0.28 & 1.00 & 1.97 & 1.26 & 0.27 & 1.00 & 1.97 & 1.26 & 0.26 & 1.00 & 1.97 \\
  \hline 
  $LB^*_{\widehat{\mathcal{P}}_{KPY}}$ &6883.86 & 860.32 & 3877 & 9275 & 1.28 & 0.28 & 1.00 & 1.97 & 1.26 & 0.27 & 1.00 & 1.97 & 1.26 & 0.26 & 1.00 & 1.97\\
  \hline
 $LB^*_{\mathcal{P}_{MJC}}$ & 2328.92 & 300.24 & 1757 & 3195 & 1.28 & 0.28 & 1.00 & 1.97 & 1.26 & 0.27 & 1.00 & 1.97 & 1.26 & 0.26 & 1.00 & 1.97\\
  \hline
  $LB^*_{\mathcal{D}_{MJC}}$ &3128.35 & 1231.49 & 1458 & 7638 & 1.28 & 0.28 & 1.00 & 1.97 & 1.26 & 0.27 & 1.00 & 1.97 & 1.26 & 0.26 & 1.00 & 1.97\\
  \hline
  $LB_{CG}$ & 87.46 & 10.67 & 53 & 117 & 1.80 & 0.38 & 1.00 & 2.00 & - & - & - & - & 1.77 & 0.38 & 1.00 & 2.00\\
  \hline
  $LB_{KZ}$ & 24.06 & 6.95 & 5 & 38 & 1.83 & 0.38 & 1.00 & 2.00& -&-&-&-& 1.80 & 0.38 & 1.00 & 2.00\\
  \hline 
\end{tabular}}
\caption{R1000-1000-0.5-0.5}
\end{table*}



\begin{table*}
\footnotesize
\resizebox{\textwidth}{!}{
\begin{tabular}{| l || r | r | r | r || r | r | r | r || r | r | r | r || r | r | r | r |}
  \hline
   & \multicolumn{4}{|c|}{time for lb (ms)} 
   & \multicolumn{4}{|c|}{Gap-medSol}
   & \multicolumn{4}{|c|}{Gap-minSol}
   & \multicolumn{4}{|c|}{Gap-Opt}\\
  \hline
   & mean & std & min & max & mean & std & min & max & mean & std & min & max & mean & std & min & max \\
   \hline
   MILP & 60.13 & 24.59 &19 &139 &-&-&-&-&-&-&-&-&-&-&-&-\\
  \hline
  $LB^*_5$ & 7.01 & 1.12 & 6 & 11 & inf & inf & 1.52 & inf & inf & inf & 1.52 & inf & inf & nan & 1.52 & inf \\
  \hline
  $LB^*_{10}$ & 14.21 & 1.31 & 13 & 21 & 1.79 & 0.24& 1.37 & 3.18 & 1.78 & 0.23 & 1.37 & 3.18 & 1.73 & 0.22 & 1.37 & 3.09	\\
  \hline  
  $LB^*_{15}$ & 24.28 & 2.05& 23 & 32 & 1.56 & 0.11 & 1.29 & 1.91 & 1.55 & 0.11 & 1.29 & 1.91 & 1.51 & 0.10& 1.29 & 1.87\\
  \hline
  $LB^*_{20}$ & 35.93 & 1.90 & 34 & 46 & 1.48 & 0.10 & 1.25 & 1.73 & 1.47 & 0.10 & 1.25 & 1.73 & 1.44 & 0.08 & 1.25 & 1.71\\
  \hline
 $LB^*_{DO}$  & 110.06 & 38.38 & 38 & 218 & 1.42 & 0.08 & 1.19 & 1.68 & 1.40 & 0.08 & 1.19 & 1.68 & 1.37 & 0.07 & 1.19 & 1.64\\
  \hline
  $LB^*_{\widehat{\mathcal{P}}_{KPY}}$ & 4.64 & 0.67 & 4 & 8 & 1.42 & 0.08 & 1.19 & 1.68 & -&-&-&-& 1.37 & 0.07 & 1.19 & 1.64\\ 
  \hline 
  $LB^*_{\mathcal{P}_{MJC}}$ & 32.20 & 12.60 & 11 & 70 & 1.42 & 0.08 & 1.19 & 1.68 & -&-&-&-& 1.37 & 0.07 & 1.19 & 1.64\\ 
  \hline 
  $LB^*_{\mathcal{D}_{MJC}}$ & 33.32 & 11.23 & 13 & 73 & 1.42 & 0.08 & 1.19 & 1.68 & -&-&-&-& 1.37 & 0.07 & 1.19 & 1.64\\ 
  \hline 
  $LB_{CG}$ & 0.50 & 0.50 & 0 & 1 & 1.91 & 0.09 & 1.66 & 2.00 & - & - & - & - & 1.85 & 0.10 & 1.63 & 2.00 \\
  \hline
 $LB_{KZ}$ & 0.20 & 0.40 & 0 & 1 & 2.00 & 0.00 & 2.00 & 2.00 & -&-&-&-& 1.94 & 0.07 & 1.72 & 2.00 \\
\hline
\end{tabular}
}
\caption{K102-1000-1.-2}
\end{table*}

\begin{table*}
\footnotesize
\resizebox{\textwidth}{!}{
\begin{tabular}{| l || r | r | r | r || r | r | r | r || r | r | r | r || r | r | r | r |}
  \hline
   & \multicolumn{4}{|c|}{time for lb (ms)} 
   & \multicolumn{4}{|c|}{Gap-medSol}
   & \multicolumn{4}{|c|}{Gap-minSol}
   & \multicolumn{4}{|c|}{Gap-Opt}\\
  \hline
   & mean & std & min & max & mean & std & min & max & mean & std & min & max & mean & std & min & max  \\
  \hline
  MILP & 592.29 & 250.74 & 335 & 2305 & - &-&-&-&-&-&-&-&-&-&-&-\\
  \hline
  $LB^*_{5}$ & 21.51 & 2.87 & 19 & 34 & inf & inf  & inf & inf   & inf  & inf  & inf & inf & inf  & inf  & inf & inf \\
  \hline
  $LB^*_{10}$ & 48.23 & 3.86 & 44 & 70 & 1.50 & 0.34 & 1.05& 2.55 & 1.48 & 0.33 & 1.03 & 2.55 & 1.43 & 0.32 & 1.03 & 2.47 \\
  \hline  
  $LB^*_{15}$ & 79.97 & 3.75 & 74 & 109 & 1.28 & 0.16 & 1.04 & 1.88 & 1.25 & 0.15 & 1.02 & 1.77 & 1.22 & 0.15 & 1.02 & 1.74\\
  \hline
 $LB^*_{20}$  & 117.36 & 5.05 &107 &147 & 1.20 & 0.13 & 1.03 & 1.81 & 1.16 & 0.12 & 1.01 & 1.81 & 1.14& 0.11 & 1.01& 1.80\\
  \hline
  $LB^*_{DO}$ & 298.43 & 125.86 & 107 & 847 & 1.11&0.05&1.02& 1.23 & 1.05&0.03&1.01&1.15 & 1.05& 0.02 & 1.01 & 1.15 \\
  \hline
  $LB^*_{\widehat{\mathcal{P}}_{KPY}}$ & 117.69 & 10.20 & 91 & 145 & 1.11 & 0.05 & 1.02 & 1.23 &-&-&-&-& 1.05 & 0.02 & 1.01 & 1.15  \\
  \hline
    $LB^*_{\mathcal{P}_{MJC}}$  & 137.13 & 42.56 & 58 & 367 & 1.11 & 0.05 & 1.02 & 1.23 &-&-&-&-& 1.05 & 0.02 & 1.01 & 1.15  \\
  \hline
   $LB^*_{\mathcal{D}_{MJC}}$  & 260.66 & 72.08 & 109 & 507 & 1.11 & 0.05 & 1.02 & 1.23 &-&-&-&-& 1.05 & 0.02 & 1.01 & 1.15  \\
  \hline
  $LB_{CG}$ & 3.62 & 0.52 & 3 & 5 & 1.64 & 0.10 & 1.42 & 1.91 & - & - & - & - & 1.56 & 0.08 & 1.40 & 1.77\\
  \hline
  $LB_{KZ}$ & 2.30 & 0.48 & 2 & 4 & 2.00 & 0.00 & 2.00 & 2.00 & -&-&-&-& 1.90 & 0.05 & 1.68 & 2.00 \\
  \hline 
\end{tabular}
}
\caption{K402-1000-1.-10}
\end{table*}

\begin{table*}
\footnotesize
\resizebox{\textwidth}{!}{
\begin{tabular}{| l || r | r | r | r || r | r | r | r || r | r | r | r || r | r | r | r |}
  \hline
   & \multicolumn{4}{|c|}{time for lb (ms)} 
   & \multicolumn{4}{|c|}{Gap-medSol}
   & \multicolumn{4}{|c|}{Gap-minSol}
   & \multicolumn{4}{|c|}{Gap-Opt} \\
  \hline
   & mean & std & min & max & mean & std & min & max & mean & std & min & max & mean & std & min & max \\
  \hline
  MILP & 4344.55 & 2043.85 & 2136 & 12849 & -&-&-&- &-&-&-&- &-&-&-&-\\
  \hline
  $LB^*_{5}$& 55.68 & 2.68 & 52 & 71 & inf & inf  & inf & inf   & inf  & inf  & inf & inf & inf  & inf  & inf & inf \\
  \hline
  $LB^*_{10}$ & 124.02 & 4.82 & 116 & 148 &  1.64 & 0.50 & 1.15 & 5.51  & 1.63 & 0.50 & 1.15 & 5.51  & 1.56 & 0.48 & 1.15 & 5.31\\
  \hline  
  $LB^*_{15}$ & 201.77 & 6.63 & 189 & 228 & 1.39 & 0.20 & 1.11 & 2.36 & 1.38 & 0.20 & 1.09 & 2.36 & 1.32 & 0.18 & 1.09 & 2.19 \\
  \hline
  $LB^*_{20}$ &291.10 & 9.70 & 276 & 319 & 1.30 &0.15 & 1.11 & 2.29 &1.29 & 0.15 & 1.11 & 2.29 & 1.23 & 0.14 & 1.08 & 2.19\\
  \hline
 $LB^*_{DO}$ & 4700.47 & 1741.47 & 1924 & 11682 &1.10 & 0.03 & 1.04 & 1.21 &1.05 & 0.02& 1.02 & 1.12 & 1.04 &  0.01 & 1.02 & 1.09\\
   \hline 
  $LB^*_{\widehat{\mathcal{P}}_{KPY}}$ & 1032.94 & 82.58 & 833 & 1321 & 1.10 & 0.03 & 1.04 & 1.21 & - & - & - & - & 1.04 & 0.01 & 1.02 & 1.09\\
  \hline
  $LB^*_{\mathcal{P}_{MJC}}$ & 1096.37 & 241.22 & 710 & 1770 & 1.10 & 0.03 & 1.04 & 1.21 & - & - & - & - & 1.04 & 0.01 & 1.02 & 1.09\\
  \hline
  $LB^*_{\mathcal{D}_{MJC}}$ & 3985.08 & 680.85 & 2283 & 5834 & 1.10 & 0.03 & 1.04 & 1.21 & - & - & - & - & 1.04 & 0.01 & 1.02 & 1.09\\
  \hline
  $LB_{CG}$ & 9.71 & 0.81 & 8 & 15 & 1.63 & 0.06 & 1.50 & 1.81 & - & - & - & - & 1.55 & 0.05 & 1.43 & 1.68 \\
  \hline
  $LB_{KZ}$ & 5.97 & 0.71 & 5 & 10	& 2.00 & 0.00 & 2.00 & 2.00 & -&-&-&-& 1.90 & 0.05 & 1.75 & 1.99 \\
 \hline 
\end{tabular}}
\caption{K1002-1000-1.-10}
\end{table*}

\begin{table*}
\footnotesize
\resizebox{\textwidth}{!}{
\begin{tabular}{| l || r | r | r | r || r | r | r | r || r | r | r | r || r | r | r | r |}
  \hline
   & \multicolumn{4}{|c|}{time for lb (ms)} 
   & \multicolumn{4}{|c|}{Gap-medSol}
   & \multicolumn{4}{|c|}{Gap-minSol}
   & \multicolumn{4}{|c|}{Gap-Opt} \\
  \hline
   & mean & std & min & max & mean & std & min & max & mean & std & min & max & mean & std & min & max \\
  \hline
  MILP & 22826.90 & 9761.98 & 7450 & 57675 &-&-&-&-&-&-&-&-&-&-&-&-\\
  \hline
  $LB^*_{5}$ &118.77 & 4.69 & 110 & 140 & inf & inf  & inf & inf   & inf  & inf  & inf & inf & inf  & inf  & inf & inf \\
  \hline
  $LB^*_{10}$ &256.66 & 9.72 & 242 & 294 & 1.72 & 0.62 & 1.20 & 5.06  & 1.71 & 0.62 & 1.20 & 5.06  & 1.63 & 0.59 &  1.16 & 4.81\\
  \hline  
  $LB^*_{15}$ & 413.50 & 11.51 & 393 & 449 & 1.42 & 0.20 & 1.15 & 2.33 & 1.42 & 0.20 &1.15 & 2.33 & 1.34 & 0.18 & 1.12 & 2.23\\
  \hline
  $LB^*_{20}$ & 588.99 & 14.33 & 558 & 629 & 1.34 & 0.15 & 1.15 & 2.01 & 1.34 & 0.15 & 1.15 & 2.01 &  1.27 & 0.14 & 1.10 & 1.89\\
  \hline
 $LB^*_{DO}$  & 52684.10 & 23175.80 & 21865 & 180398 & 1.10 & 0.02 & 1.06 & 1.16 & 1.05 & 0.01 & 1.02 & 1.09 & 1.04 & 0.01 & 1.02 & 1.06\\
  \hline 
  $LB^*_{\widehat{\mathcal{P}}_{KPY}}$ & 2569.08 & 359.57 & 1966 & 3913 & 1.10 & 0.02 & 1.06 & 1.16 & - & - & - & - & 1.04 & 0.01 & 1.02 & 1.06 \\
  \hline
   $LB^*_{\mathcal{P}_{MJC}}$ & 6159.52 & 1311.35 & 3953 & 9760 & 1.10 & 0.02 & 1.06 & 1.16 & - & - & - & - & 1.04 & 0.01 & 1.02 & 1.06 \\
  \hline
   $LB^*_{\mathcal{D}_{MJC}}$ & 34758.00 & 5787.99 & 21827 & 53184 & 1.10 & 0.02 & 1.06 & 1.16 & - & - & - & - & 1.04 & 0.01 & 1.02 & 1.06 \\
  \hline
  $LB_{CG}$ & 20.39 & 0.95 & 19 & 26 & 1.62 & 0.04 & 1.54 & 1.73 & - & - & - & - & 1.54 & 0.04 & 1.45 & 1.67 \\
  \hline
  $LB_{KZ}$ & 12.23 & 0.65 & 11 & 15 & 2.00 & 0.00 & 2.00 & 2.00 & -&-&-&-& 1.90 & 0.03 & 1.81 & 1.97 \\
  \hline
\end{tabular}}
\caption{K2002-1000-1.-10}
\end{table*}

We have performed the same experiments on two valued graphs representing real cities (the edge values represent distances):
\begin{itemize}
\item a graph representing New York city with 264346 nodes and 733846 edges,
\item a graph representing the San Fransisco Bay with 321270 nodes and 800172 edges.
\end{itemize}

Those graphs are available on the website of the $9^{th}$ DIMACS implementation challenge on the computation of shortest paths\footnote{\url{http://www.dis.uniroma1.it/challenge9/download.shtml}}. For these two graphs, for each instance, the starting node and the destination node are randomly drawn according to a uniform distribution. For each cost $c_{ij}$, the lower  bound $\underline{c}_{ij}$ (resp. upper bound $\overline{c}_{ij}$) is randomly drawn according to a uniform distribution in interval $[c_{ij} - c_{ij}/10, c_{ij} ]$ (resp. $[c_{ij} , c_{ij} + c_{ij}/10]$). The results are presented in Tables 9 and 10. For those two graphs, the computation of $LB^*_{DO}$ is much faster than: the computation of $LB^*_{\widehat{\mathcal{P}}_{KPY}}$ (twice faster for New York and 9 times faster for San Francisco Bay), the computation of $LB^*_{ \mathcal{D}_{MJC}}$ (5 times faster for New York and $2.5$ times faster for San Francisco Bay), and the computation of $LB^*_{ \mathcal{P}_{MJC}}$ (5 times faster for New York and $3.2$ times faster for San Francisco Bay). 

Regarding the accuracy of the considered lower bounds, one can observe that the lower bound obtained after $20$ iterations of the double oracle algorithm is close to $LB^*$ and provides a much better approximation guarantee than $LB_{KZ}$ or $LB_{CG}$. \\

\begin{table*}
\footnotesize
\resizebox{\textwidth}{!}{
\begin{tabular}{| l || r | r | r | r || r | r | r | r || r | r | r | r |}
  \hline
   & \multicolumn{4}{|c|}{time for lb (ms)} 
   & \multicolumn{4}{|c|}{Gap-medSol}
   & \multicolumn{4}{|c|}{Gap-minSol}\\
  \hline
   & mean & std & min & max & mean & std & min & max & mean & std & min & max \\
  \hline
  $LB^*_{5}$ & 1529.98 & 1077.65 & 47 & 3643 & 3.03 & 8.33 & 1.12 & 82.80  & 3.03 & 8.33 & 1.09 & 82.80\\
  \hline
  $LB^*_{10}$ & 3393.42 & 2359.67 & 103 & 8095 & 1.40 & 0.20 & 1.06 & 2.27 &  1.39 & 0.20 & 1.06 & 2.25\\
  \hline  
  $LB^*_{15}$ & 5752.06 & 3963.95 & 149 & 13300 & 1.30 & 0.14 & 1.05 & 1.77 & 1.29 & 0.15 & 1.05 & 1.77\\
  \hline
  $LB^*_{20}$ & 8083.67 & 5669.42 & 141 & 19378 & 1.27 & 0.12 & 1.05 & 1.66 & 1.26 & 0.12 & 1.05 & 1.61\\
  \hline
  $LB^*_{DO}$  & 48716.95 & 64774.72 & 139 & 333969 & 1.23 & 0.12 & 1.05 & 1.62 & 1.22 & 0.12 & 1.04 & 1.58\\
  \hline 
  $LB^*_{\widehat{\mathcal{P}}_{KPY}}$ & 107127.28 & 37887.70 & 12683 & 230967 & 1.23 & 0.12 & 1.05 & 1.62 & - & - & - & - \\
  \hline
  $LB^*_{\mathcal{P}_{MJC}}$ & 262040.03 & 278177.71 & 6284 & 1398590 & 1.23 & 0.12 & 1.05 & 1.62 & - & - & - & - \\
  \hline
  $LB^*_{\mathcal{D}_{MJC}}$ & 253150.12 & 249415.21 & 17698 & 1315000 & 1.23 & 0.12 & 1.05 & 1.62 & - & - & - & - \\
  \hline
  $LB_{CG}$ & 658.59 & 119.41 & 494 & 987 & 1.96 & 0.05 & 1.73 & 2.00 & - & - & - & -  \\
  \hline
  $LB_{KZ}$ & 248.71 & 176.25 & 7 & 598 & 2.00 & 0.00 & 2.00 & 2.00 & -&-&-&-\\
 \hline 
\end{tabular}}
\caption{New York City}
\end{table*}

\begin{table*}
\footnotesize
\resizebox{\textwidth}{!}{
\begin{tabular}{| r || r | r | r | r || r | r| r | r || r | r | r | r|}
  \hline
   & \multicolumn{4}{|c|}{time for lb (ms)} 
   & \multicolumn{4}{|c|}{Gap-medSol}
   & \multicolumn{4}{|c|}{Gap-minSol}\\
   \hline
   & mean & std & min & max & mean & std & min & max & mean & std & min & max \\
  \hline
  $LB^*_{5}$ & 2243.13 & 1482.35 & 26 & 5773 & 1.70 & 0.43 & 1.00 & 3.37 & 1.69 & 0.43 & 1.00 & 3.34\\
  \hline
  $LB^*_{10}$ & 4717.81 & 3290.32 & 26 & 12826 & 1.51 & 0.25 & 1.00 & 2.07 & 1.49 & 0.25 & 1.00 & 2.07\\
  \hline  
  $LB^*_{15}$ & 6693.88 & 5134.00 & 24 & 19480 & 1.47 & 0.23 & 1.00 & 1.89 & 1.46 & 0.23 & 1.00 & 1.89\\
  \hline
  $LB^*_{20}$ & 8278.93 & 7044.32 & 26 & 27046 & 1.47 & 0.23 & 1.00 & 1.89 & 1.45 & 0.23 &  1.00 & 1.89\\
  \hline
  $LB^*_{DO}$  & 11580.90 & 14054.12 & 29 & 75825 & 1.46 &  0.23 & 1.00 & 1.89 & 1.45 & 0.23 & 1.00 & 1.89\\
  \hline
  $LB^*_{\widehat{\mathcal{P}}_{KPY}}$ & 90649.60 & 34228.70 & 12531 & 250604 & 1.46 &  0.23 & 1.00 & 1.89 & - & - & - & -\\
  \hline
   $LB^*_{\mathcal{P}_{MJC}}$ & 37484.20 & 27275.40 & 6586 & 182457 & 1.46 &  0.23 & 1.00 & 1.89 & - & - & - & -\\
  \hline
   $LB^*_{\mathcal{D}_{MJC}}$ & 29387.00 & 14808.80 & 11616 & 95791 & 1.46 &  0.23 & 1.00 & 1.89 & - & - & - & -\\
  \hline
  $LB_{CG}$ & 817.32 & 143.97 & 608 & 1397 & 1.94 & 0.15 & 1.00 & 2.00 & - & - & - & -\\
  \hline
  $LB_{KZ}$ & 357.48 & 228.77 & 6 & 896 & 1.98 & 0.14 & 1.00 & 2.00 & -&-&-&-\\
  \hline 
 \end{tabular}}
\caption{San Francisco Bay}
\end{table*}

\textbf{Second Experiment.} In a second experiment, we compare the performances of four exact solution approaches: three branch and bound algorithms and the approach consisting in solving $\mathcal{P}_{KPY}$ (denoted MILP in the following). The three branch and bound algorithms use a best-first exploration strategy of the search tree generated by the branching procedure designed by Montemanni et al. \cite{DBLP:journals/orl/MontemanniGD04}. They only differ by the lower bound used. Algorithm $BB_{MGD}$ uses $LB_{MGD}$, algorithm $BB_{CG}$ uses $LB_{CG}$ and algorithm $BB^*$ uses $LB^*_{DO}$.

For small instances of each graph class, we compare the time needed for each algorithm to solve the problem to optimality and the number of nodes expanded in the branch and bound tree. The solution times are measured in milliseconds. The results, presented in Figures \ref{fig:Rgraphs} and \ref{fig:Kgraphs}, are averaged over 100 randomly generated graphs. For the computation of $LB^*_{DO}$ the number of paths in the restricted game is upper bounded by 50. Put another way, if 50 paths are generated in the restricted game, then the lower bounding procedure is stopped and the current lower bound is returned (note that the number of iterations is not upper bounded by 50, as the number of generated scenarios is not). 
Note that the curves are represented in a logarithmic scale on the $y$-axis for the graphs giving the evolution of the solution time (left side of Figures \ref{fig:Rgraphs} and \ref{fig:Kgraphs}). 

For R-graphs we see that all three branch and bound algorithms outperform the algorithm solving the MILP, except for $BB_{MGD}$ that performs very badly if the cost variability is too large. This can be seen in Figure 2.d) where the number of expanded nodes explodes for high cost variabilities. Except for very low cost variabilities, $BB^*$ performs best on average thanks to the very low number of expanded nodes.
 
For K-graphs we see that all three branch and bound algorithms are outperformed by the algorithm solving the MILP, except when the layer size is large. Indeed, as already noted by Chassein and Goerigk, the number of edges in the graph increases with the layer size, which impacts badly on the computation times for the MILP. Regarding the way the branch and bounds compare themselves, usually $BB_{MGD}$ performs least except when the layer size is large, and $BB_{CG}$ performs slightly better than $BB^*$. As seen in the previous experiment, $LB^*_{DO}$ is indeed not favored by K-graphs due to the possible large number of interesting paths between $s$ and $t$.
  
\begin{figure*}[htp]
  \centering
  \resizebox{\textwidth}{!}{
  \stackunder[5pt]{ \label{fig:Ra} \includegraphics[]{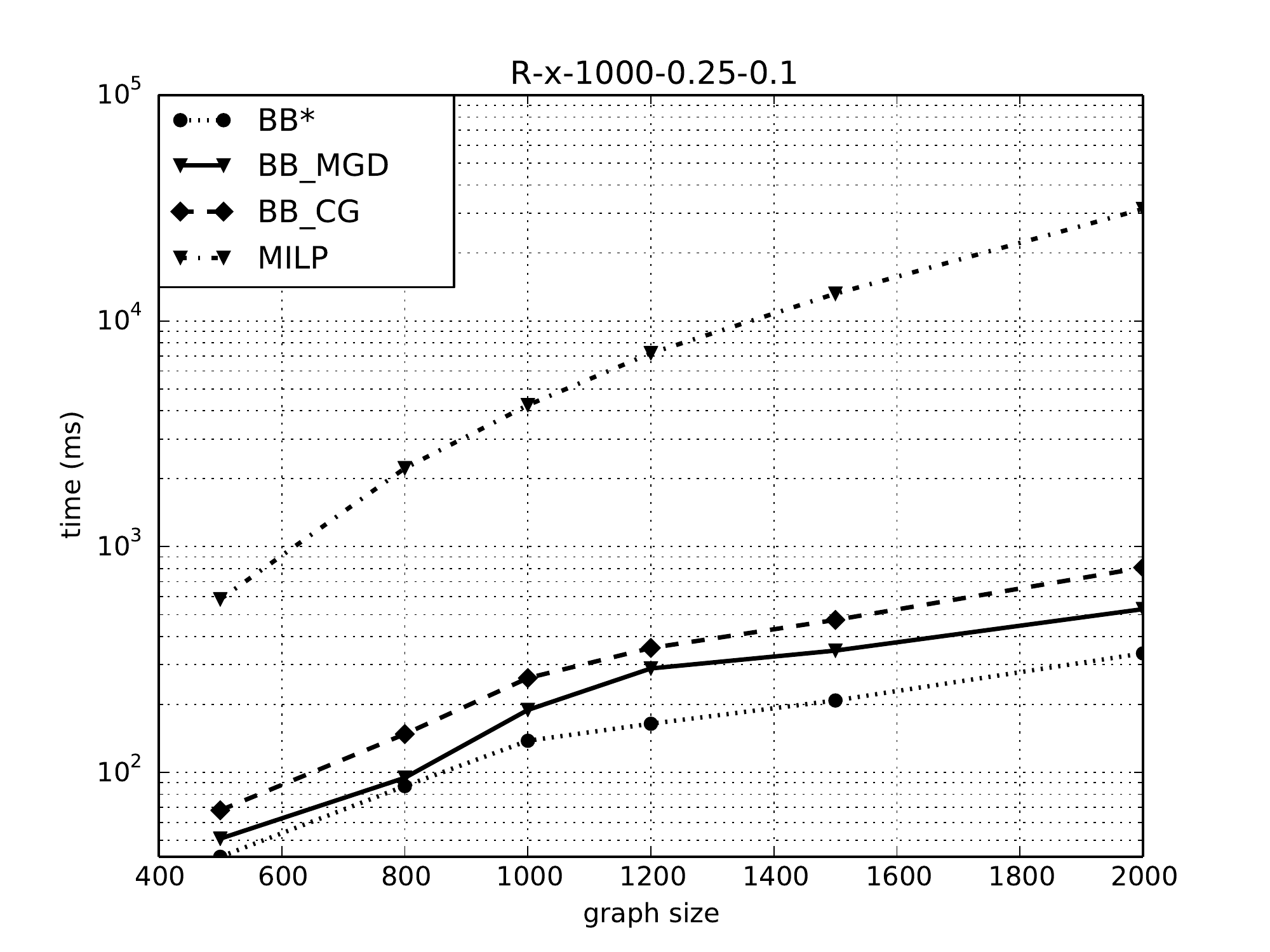}}{\Huge a) R-Graphs with different numbers of nodes}
   \hspace{5pt}
\stackunder[5pt]{ \label{fig:Rb} \includegraphics[]{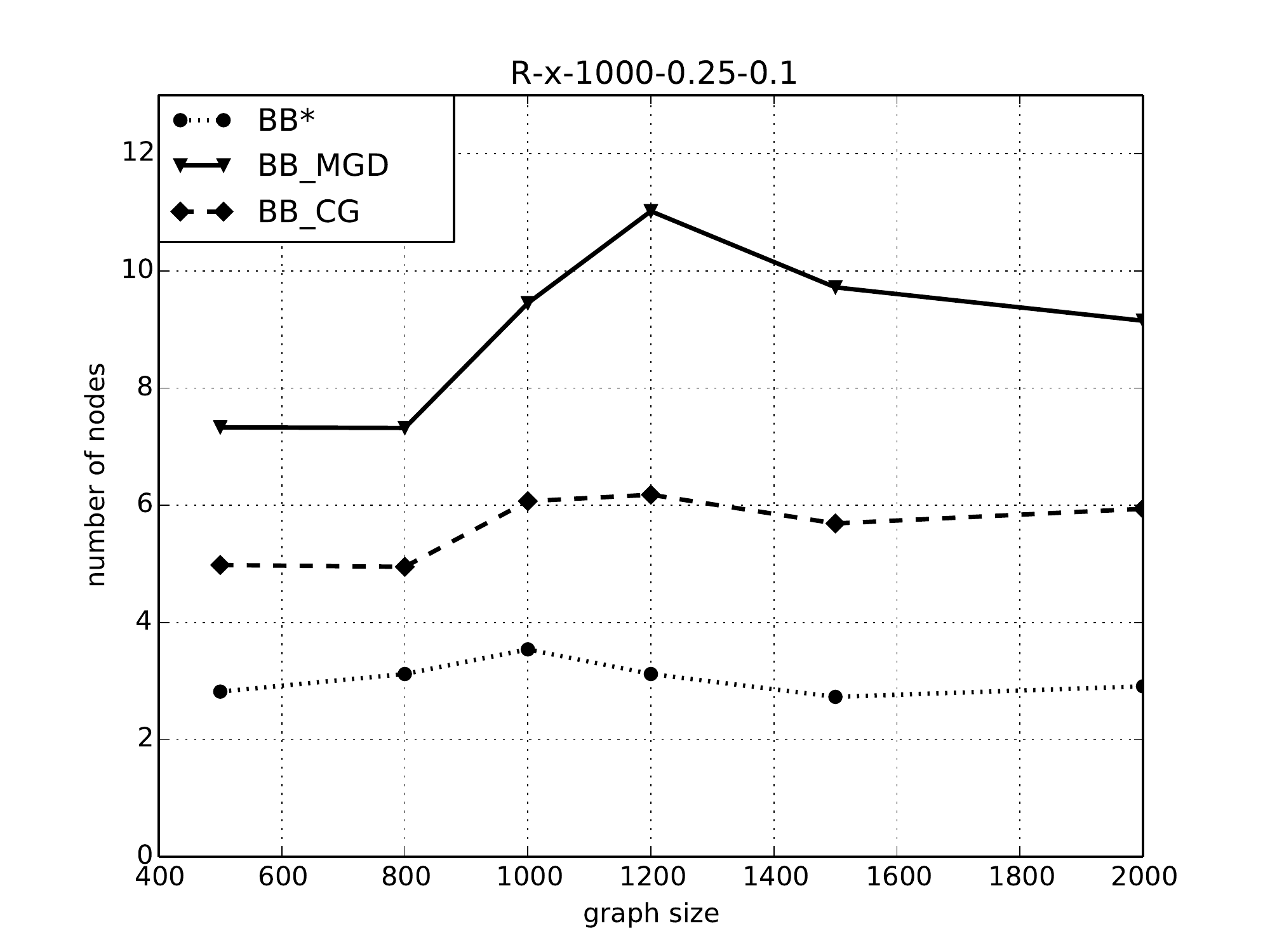}}{\Huge b) R-Graphs with different numbers of nodes}

  }
  \resizebox{\textwidth}{!}{
  \stackunder[5pt]{ \label{fig:Rc} \includegraphics[]{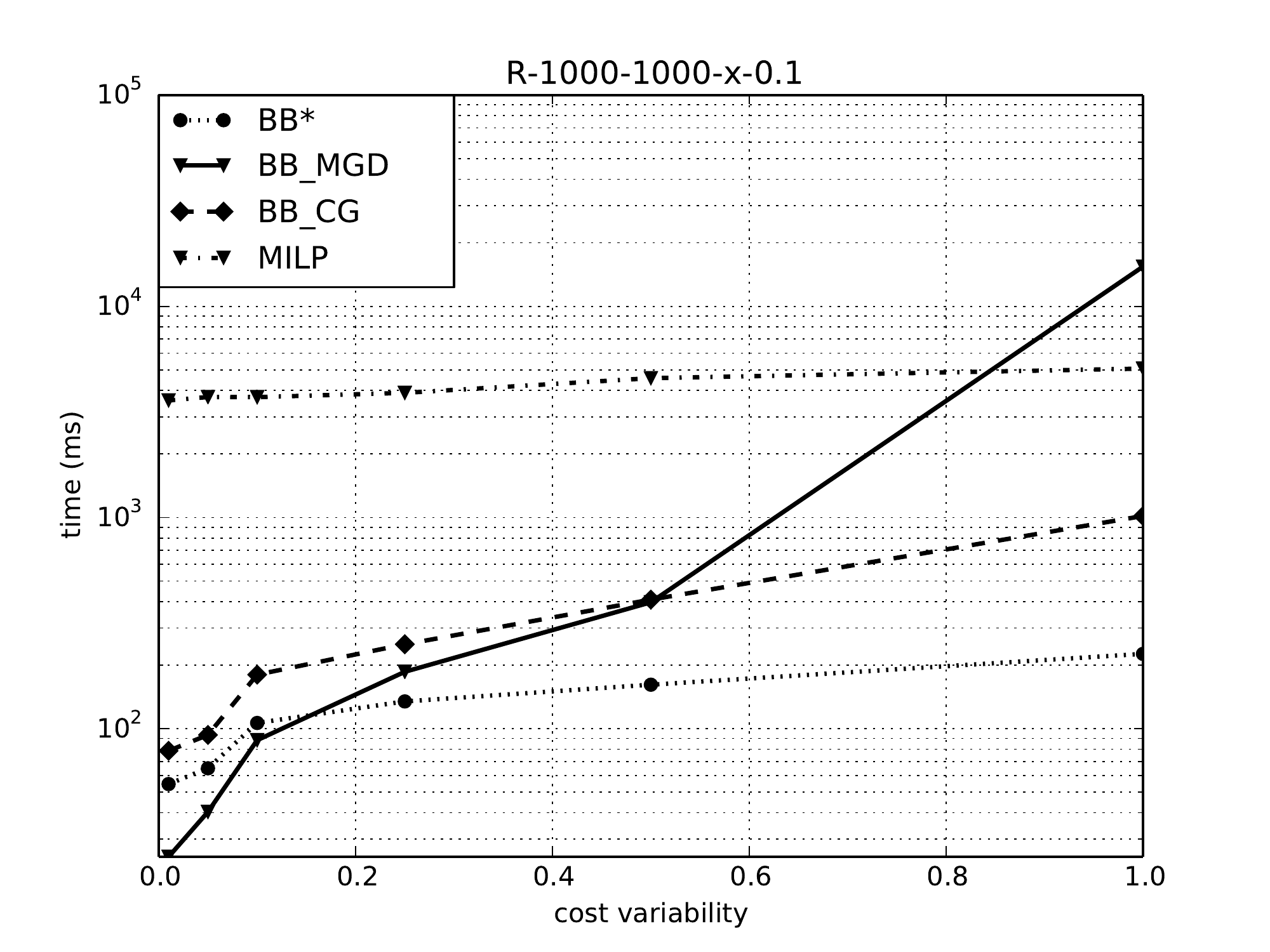}}{\Huge c) R-Graphs with different cost variabilities}
  \hspace{5pt}
  \stackunder[5pt]{ \label{fig:Rd} \includegraphics[]{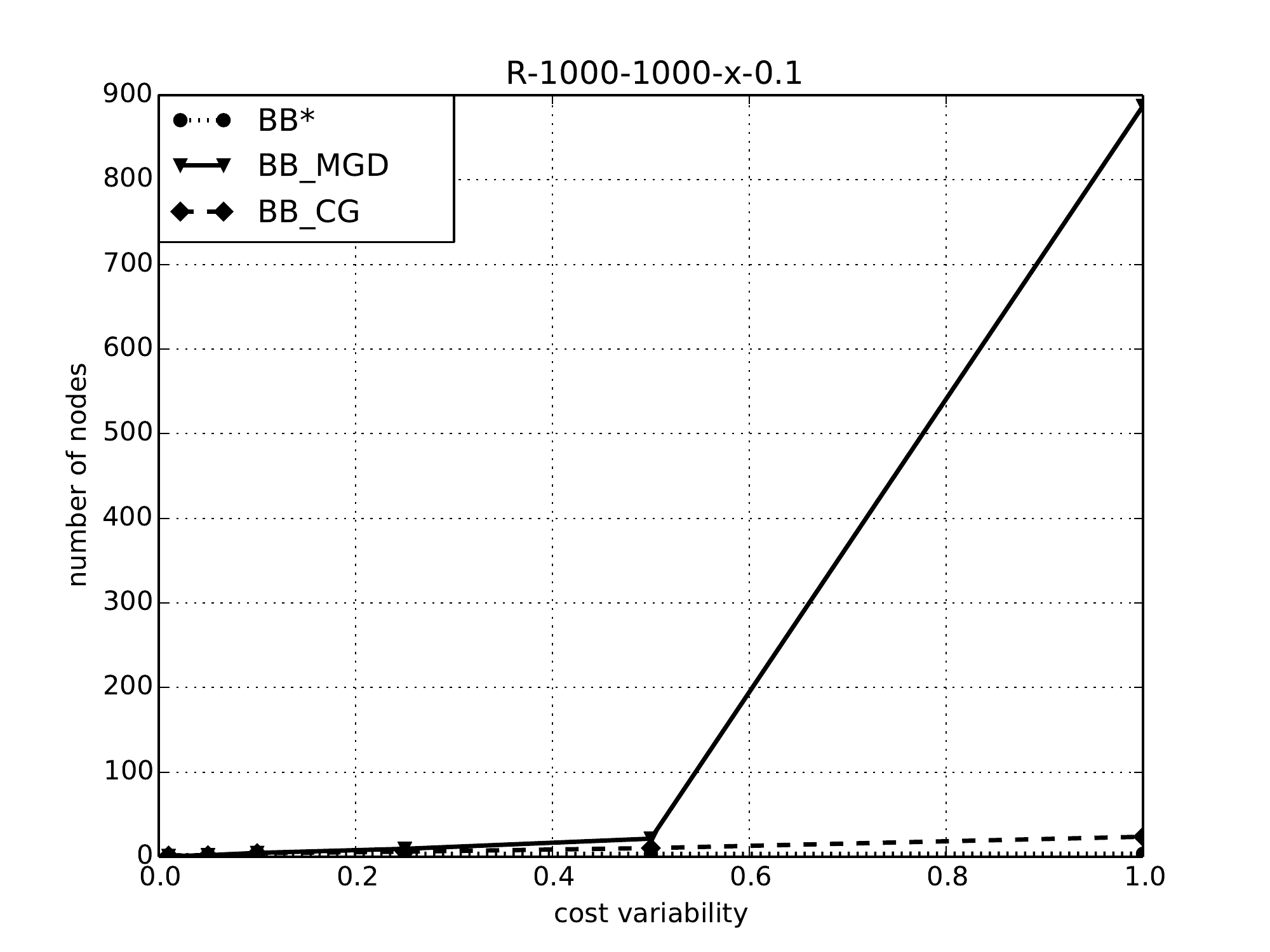}}{\Huge d) R-Graphs with different cost variabilities}
  }  
  \resizebox{\textwidth}{!}{
    \stackunder[5pt]{  \includegraphics[]{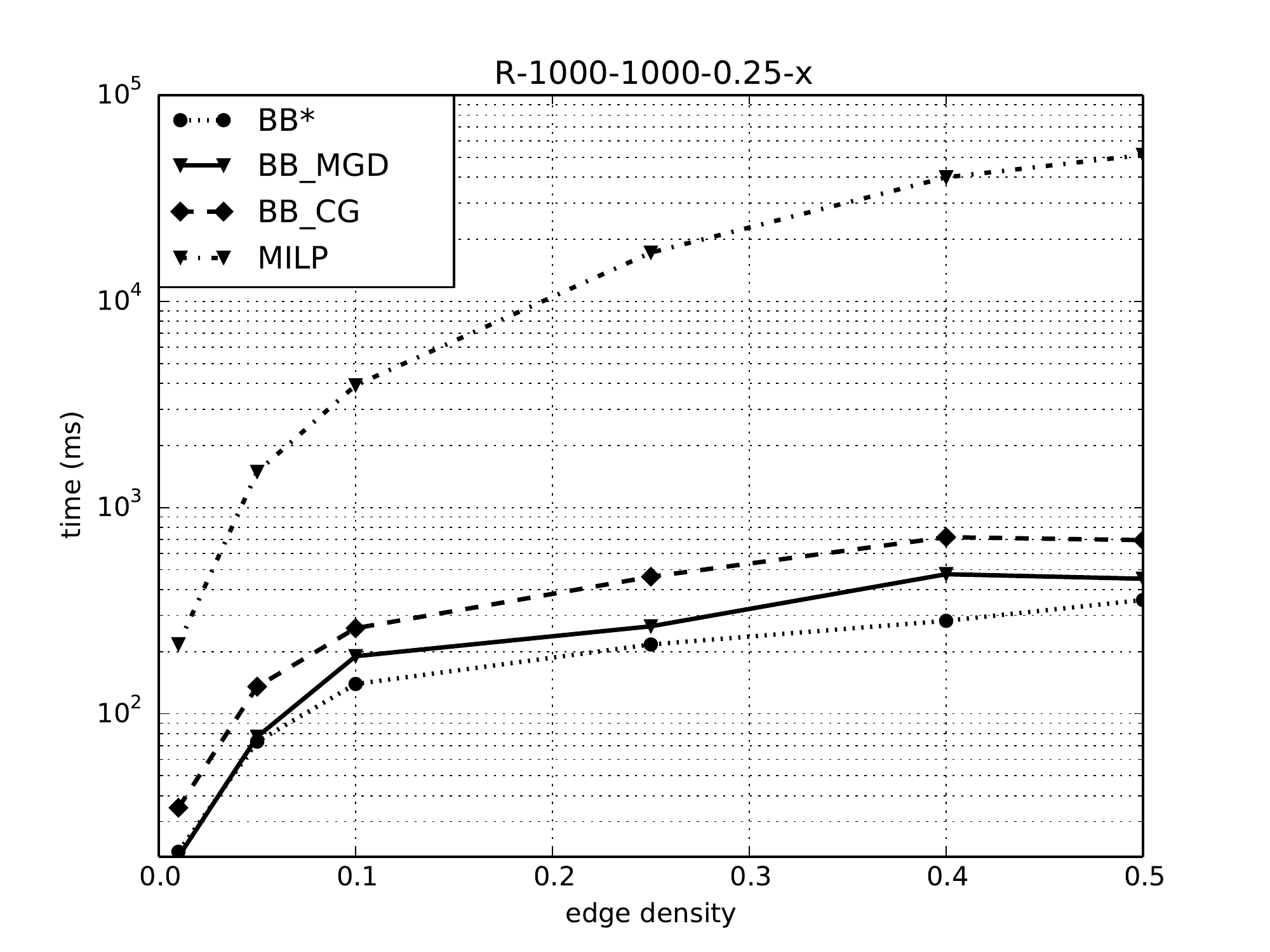} \label{fig:Re}}{\Huge e) R-Graphs with different edge densities}
  \hspace{5pt}
    \stackunder[5pt]{ \label{fig:Rf} \includegraphics[]{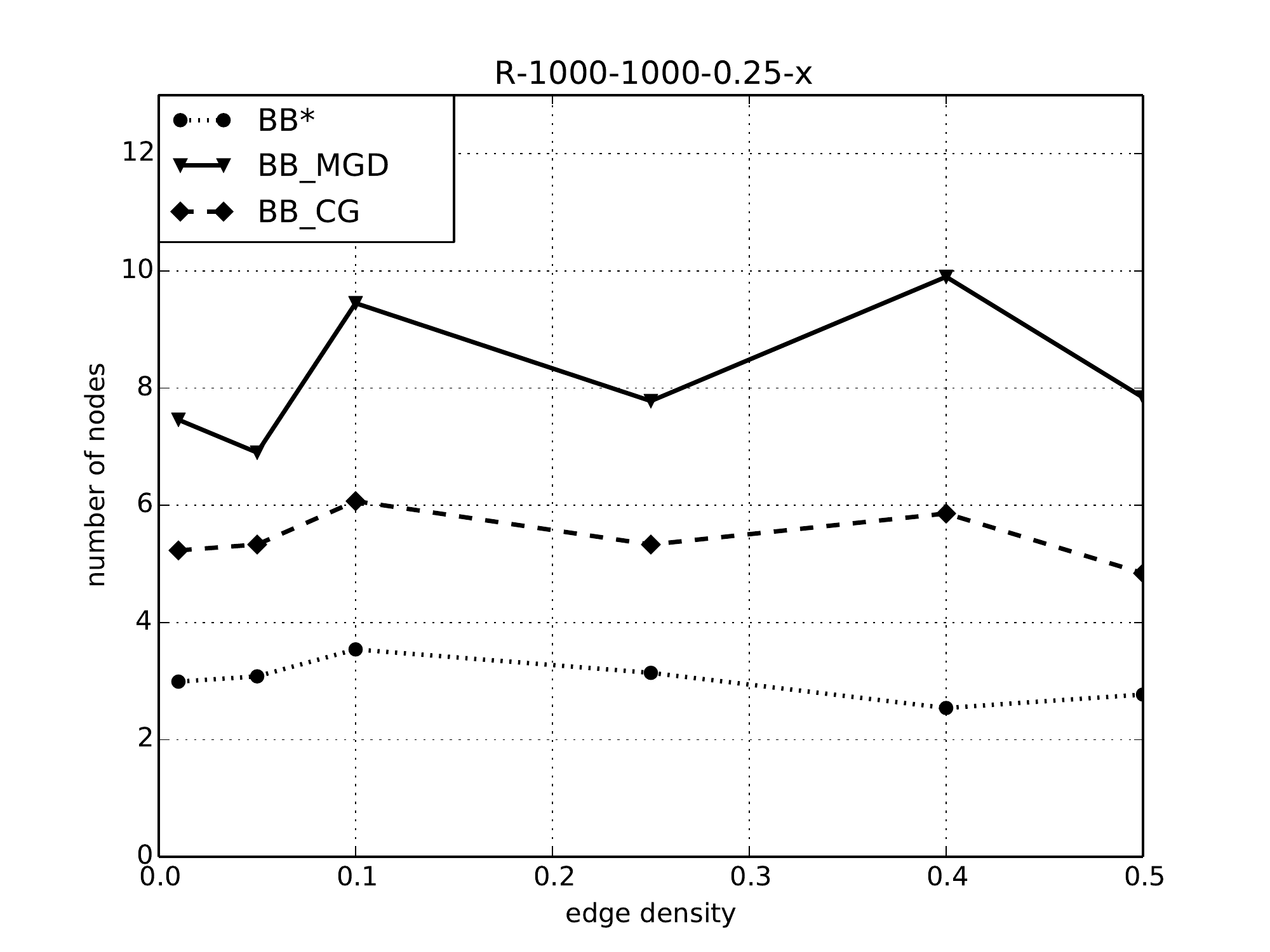}}{\Huge f) R-Graphs with different edge densities}

  }
  
  \caption{Average computation time and number of nodes of the branch and bound tree for R-graphs}
  \label{fig:Rgraphs}
\end{figure*}

\begin{figure*}[htp]
  \centering
  \resizebox{\textwidth}{!}{
  \stackunder[5pt]{ \label{fig:Ka} \includegraphics[]{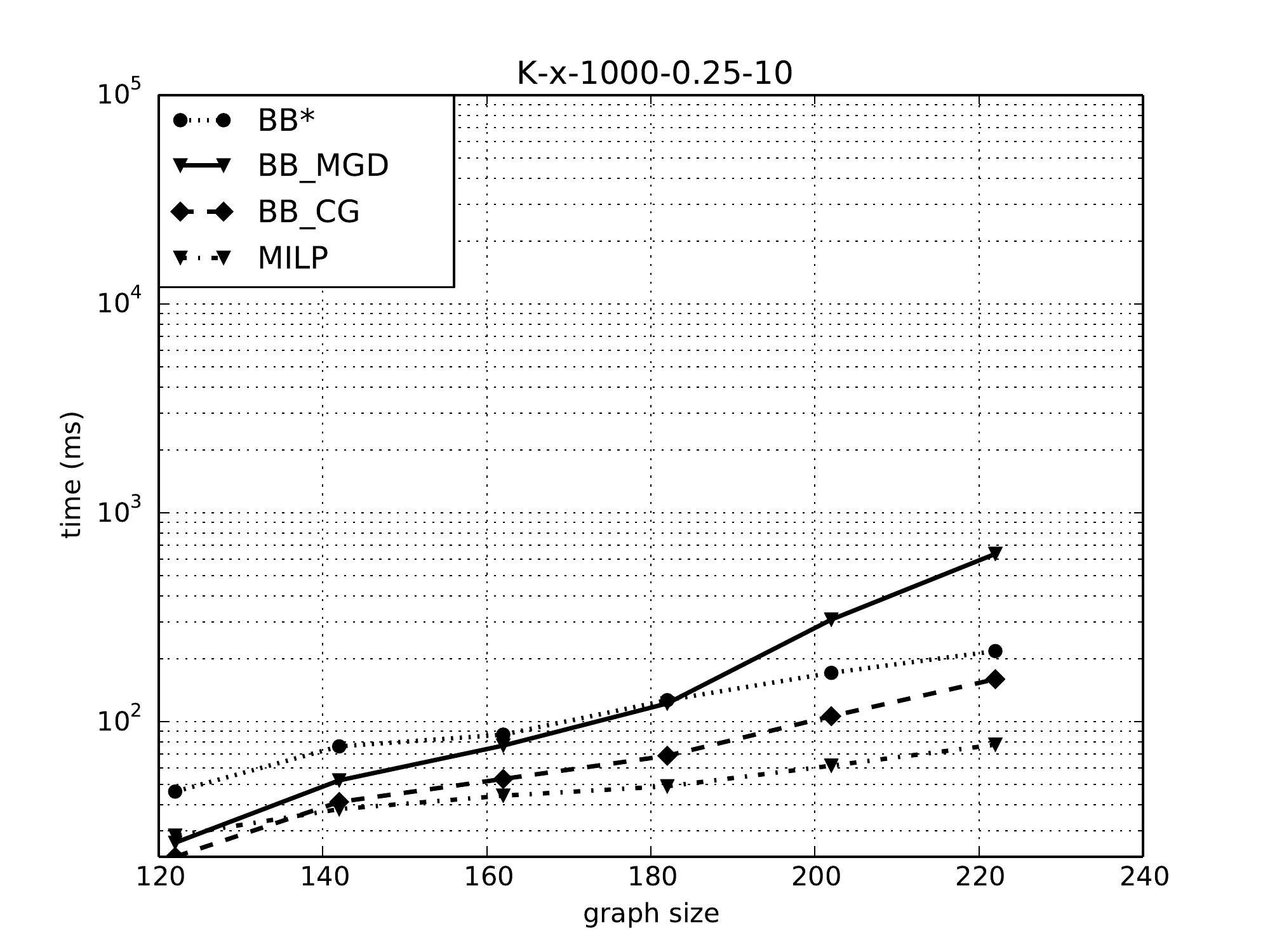}}{\Huge a) K-Graphs with different numbers of nodes}
   \hspace{5pt}
\stackunder[5pt]{ \label{fig:Kb} \includegraphics[]{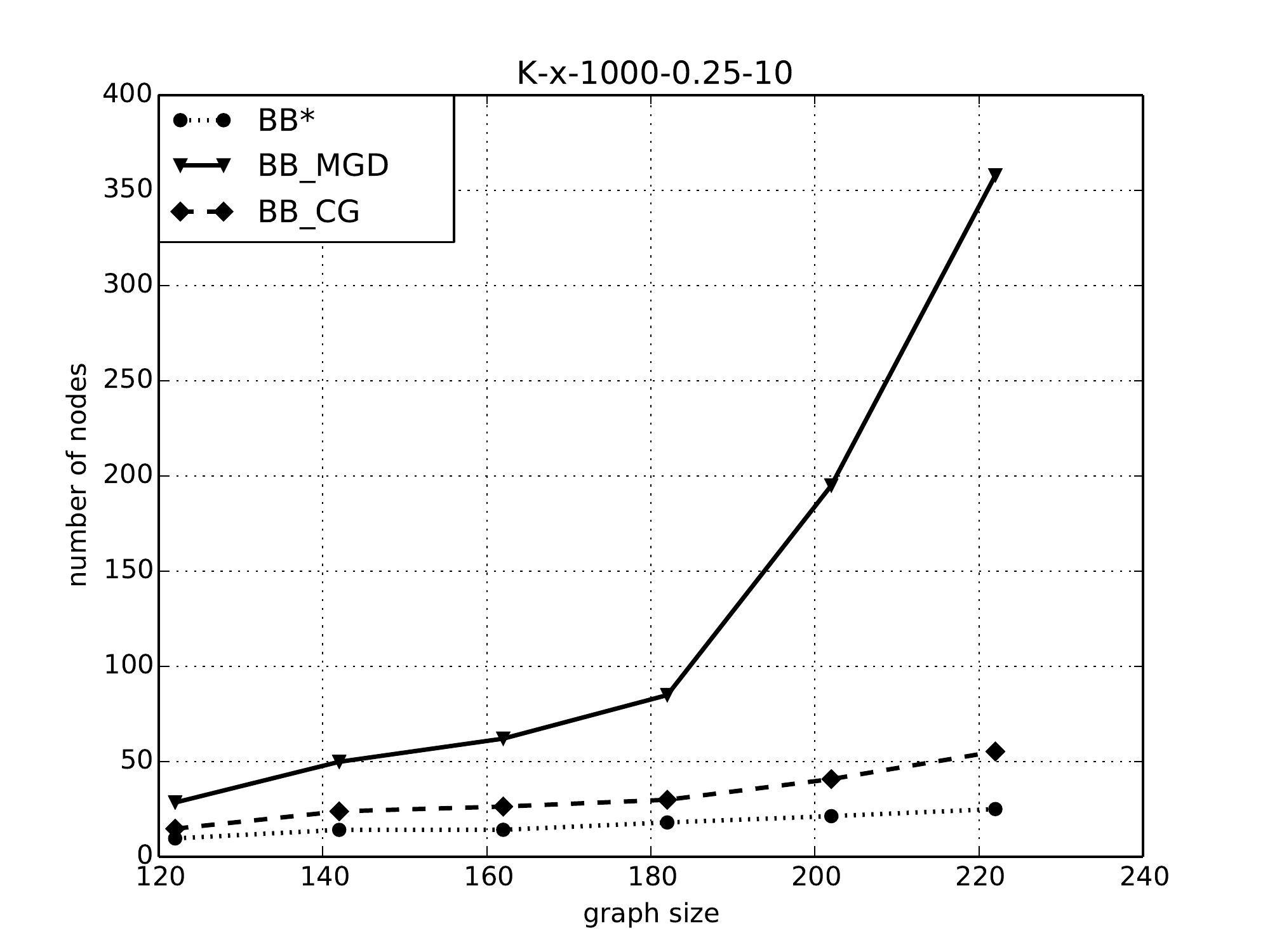}}{\Huge b) K-Graphs with different numbers of nodes}
  }
  \resizebox{\textwidth}{!}{
  \stackunder[5pt]{ \label{fig:Kc} \includegraphics[]{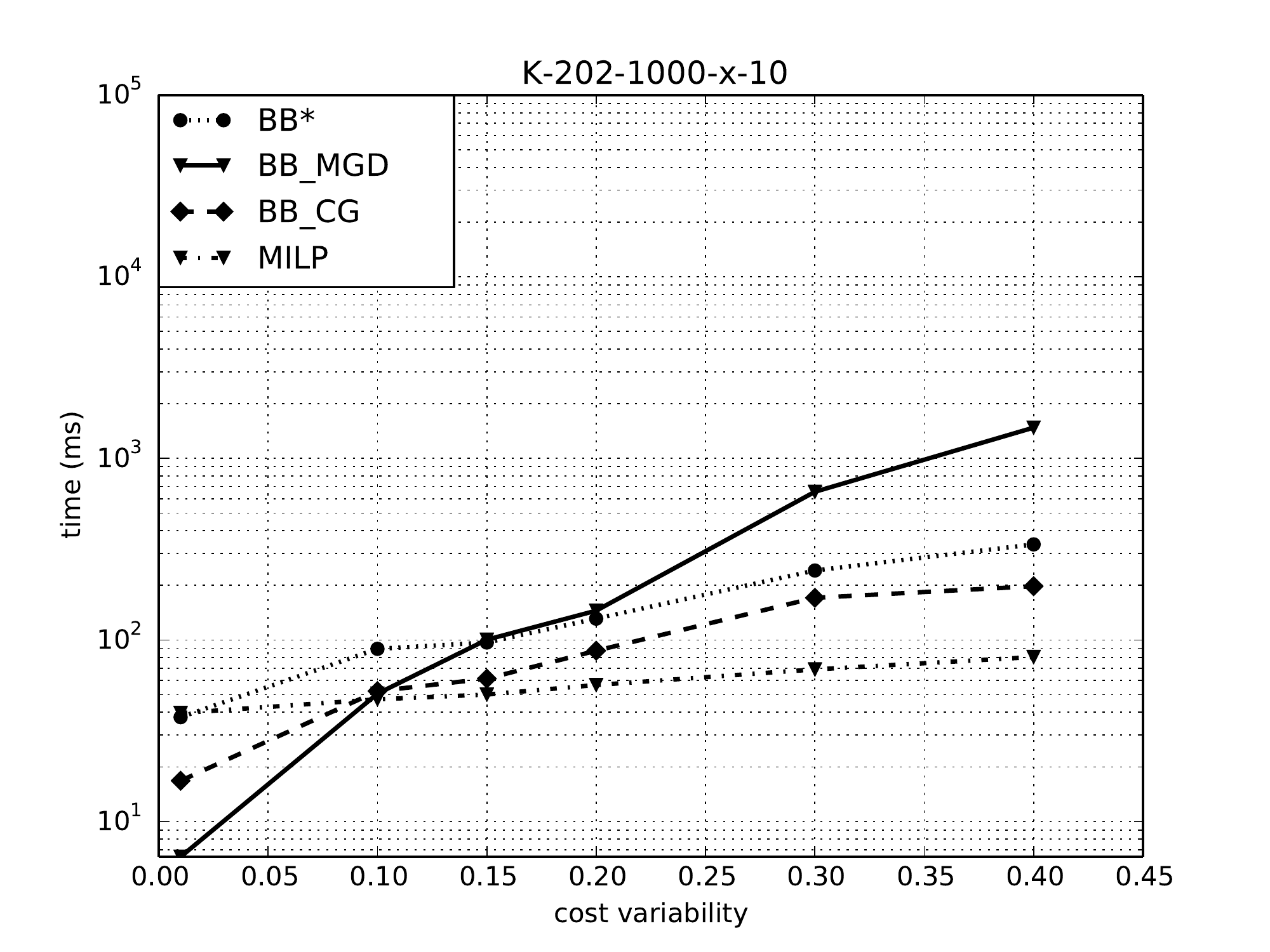}}{\Huge c) K-Graphs with different cost variabilities}
  \hspace{5pt}
  \stackunder[5pt]{ \label{fig:Kd} \includegraphics[]{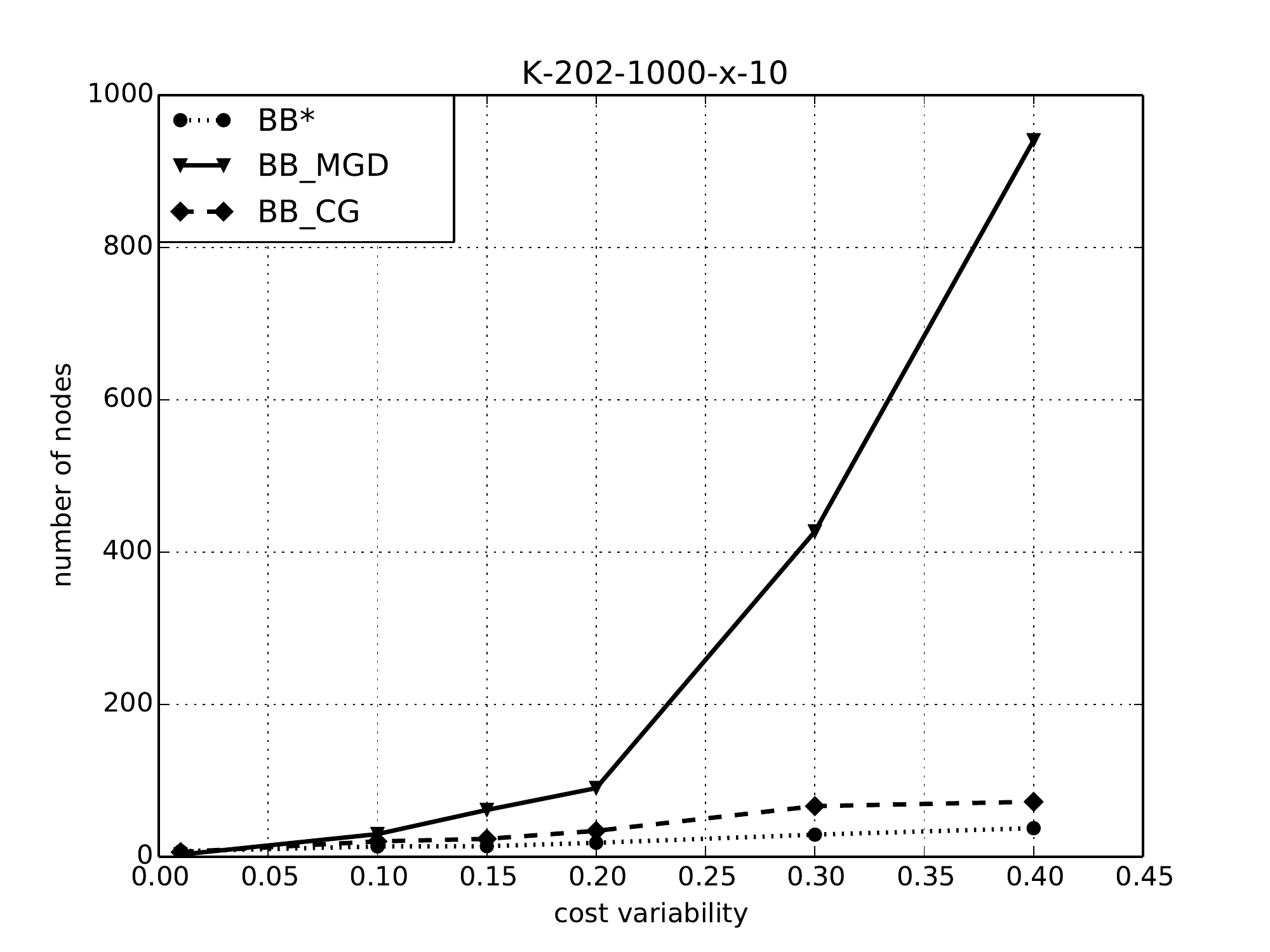}}{\Huge d) K-Graphs with different cost variabilities}
  }  
  \resizebox{\textwidth}{!}{
    \stackunder[5pt]{ \label{fig:Ke} \includegraphics[]{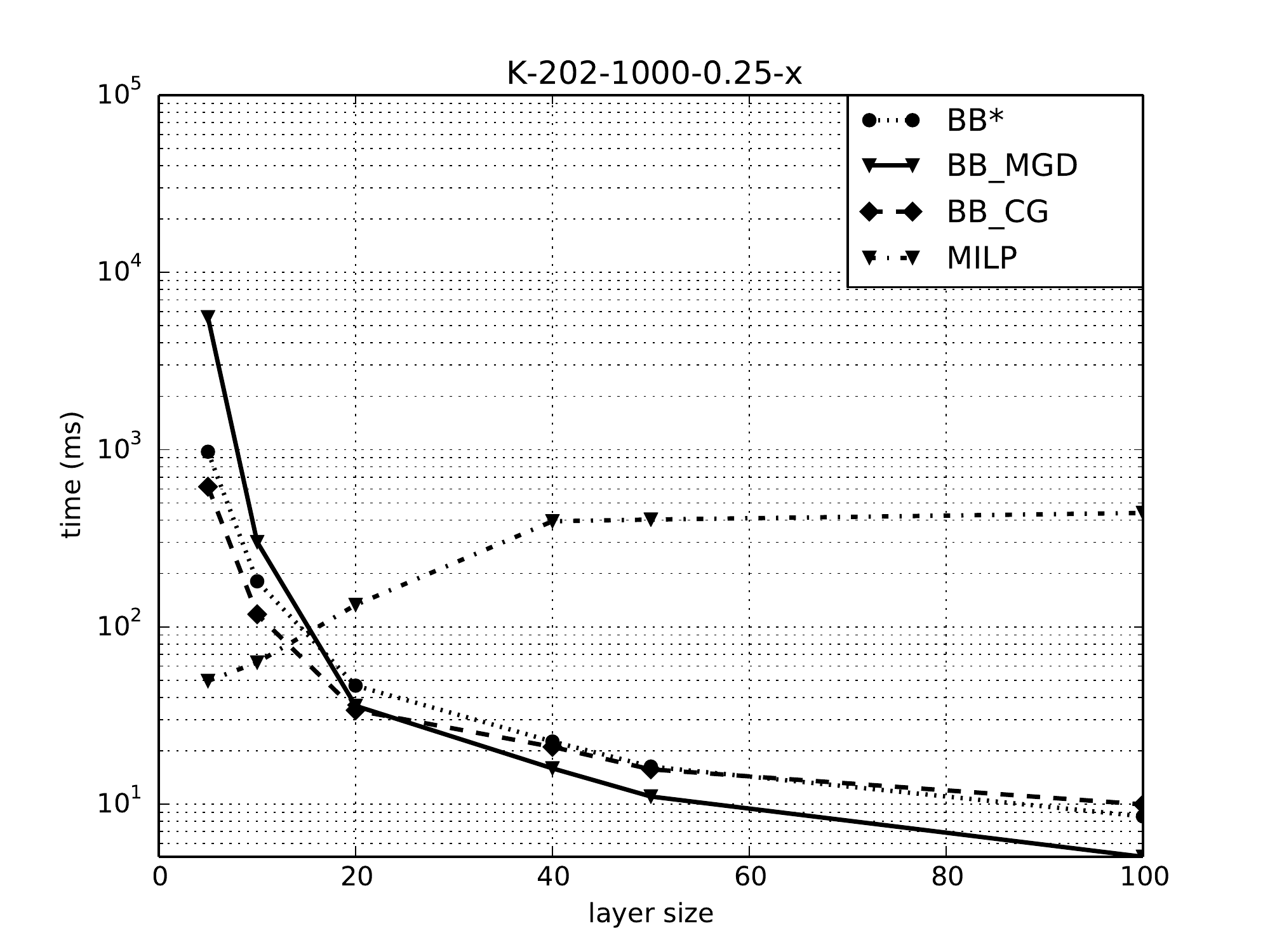}}{\Huge e) K-Graphs with different layer sizes}
  \hspace{5pt}
    \stackunder[5pt]{ \label{fig:Kf} \includegraphics[]{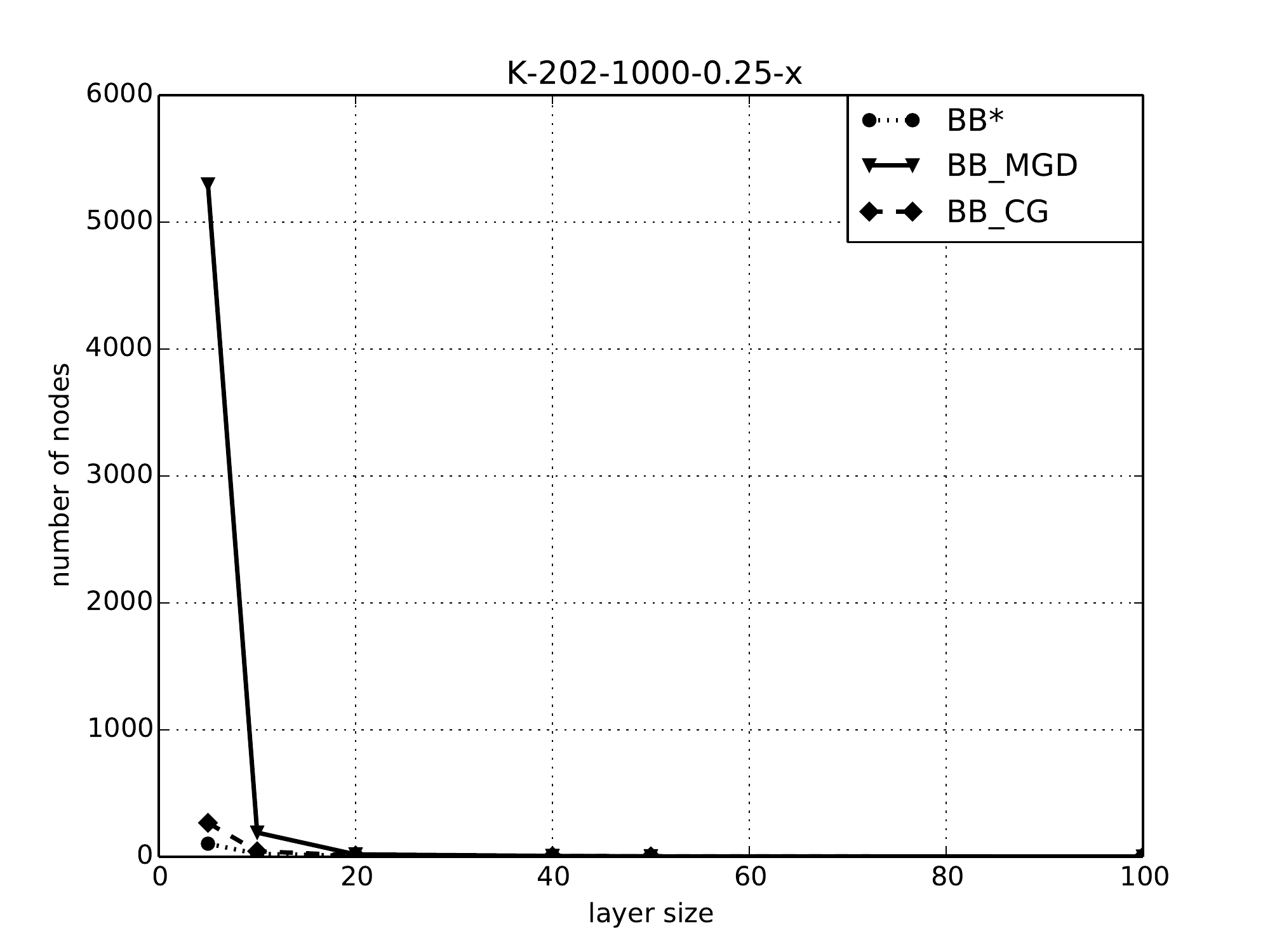}}{\Huge f) K-Graphs with different layer sizes}

  }
  
  \caption{Average computation time and number of nodes of the branch and bound tree for K-graphs}
  \label{fig:Kgraphs}
\end{figure*}

  

We conclude our experiments by comparing the three branch and bound algorithms on ``large'' R-graph instances. The results are presented in Table~\ref{Tab:LR}. Algorithm $BB^*$ explores very few nodes and outperforms the two other branch and bound algorithms on the three types of graphs considered. This experiment confirms that $BB_{MGD}$ is very sensitive to cost variability explaining its poor behavior on graphs with high cost variability. Finally, $BB_{CG}$ is more penalized by the large size of the graph. This is easily explained by the fact that Suurballe's algorithm (called many times in $BB_{CG}$) requires to run a single-source \emph{all-target} Dijkstra's algorithm while only single-source \emph{single-target} Dijkstra's algorithms are required in $BB^*$ and $BB_{MGD}$. 

\begin{table*}
\footnotesize
\resizebox{\textwidth}{!}{\begin{tabular}{| l || r | r | r || r | r | r || r | r | r|}
  \hline
   & \multicolumn{3}{|c|}{R-10000-1000-0.5-0.001} & \multicolumn{3}{|c|}{R-10000-1000-1.-0.001} & \multicolumn{3}{|c|}{R-10000-1000-0.5-0.1}\\
   \hline
   & $BB^*$ & $BB_{CG}$ & $BB_{MGD}$ & $BB^*$ & $BB_{CG}$ & $BB_{MGD}$ & $BB^*$ & $BB_{CG}$ & $BB_{MGD}$ \\
  \hline
  time (ms) & 118.29 & 1006.06 & 1017.89 & 259.47 & 3336.59 & 39430.90 & 17536.20 & 187045.00 & 51615.70 \\
  \hline
  number of nodes & 4.42 & 16.16 & 77.16 & 7.17 & 48.13 & 3037.98 & 4.31 & 12.94 & 58.93\\
  \hline  
\end{tabular}}
\caption{Large R-graphs}
\label{Tab:LR}
\end{table*}

\newpage
\section{Conclusion} 
\label{sec:conclusion}

In this paper, based on a game-theoretic view of robust optimization similar to the one proposed by Mastin et al. \cite{mastin2015randomized}, we derived a general anytime double oracle algorithm to compute an accurate lower bound on the minmax regret value of a robust optimization problem with interval data. We discussed how this lower bound can be efficiently used in a branch and bound algorithm to compute a minmax regret solution, and we provided experimental results on the robust shortest path problem. Compared to other approaches proposed in the literature, our approach leads to a significant improvement of the computation times on many instances.

For future works, a straightforward research direction would be to investigate the benefits of using the lower bounding procedure proposed here within solution algorithms for other minmax regret optimization problems. Furthermore, it could be worth investigating alternative approaches to solve the game in the lower bounding procedure. Even if the double oracle algorithm performs well in practice, no upper bound on the number of iterations has been established in the literature (other than the total number of pure strategies in the game). In this concern, Zinkevich et al. \cite{ZinkevichBB07} proposed an alternative approach to solve massive games with a theoretical upper bound on the number of generated strategies. It could be interesting to study the practical performances obtained for the lower bounding procedure if one substitutes the double oracle algorithm used here by their approach. Lastly, as emphasized by Chassein and Goerigk \cite{DBLP:journals/eor/ChasseinG15}, note that any probability distribution over $\mathcal{U}$ can be used to compute a lower bound for the minmax regret. Their bound is obtained by using a very limited set of probability distributions over $\mathcal{U}$. On the contrary, our lower bound relies on the entire set of probability distributions over $\mathcal{U}$. One research direction could be to investigate ``intermediate'' sets of probability distribution over $\mathcal{U}$ (i.e., wider than the set used by Chassein and Goerigk, but tighter than the set we use), that would still provide a good accuracy of the lower bound, together with a strong algorithmic efficiency of the lower bounding procedure.

\small
\section*{References}
\bibliographystyle{named}
\bibliography{ijcai15_209}

\end{document}